\documentclass[conference,compsoc]{IEEEtran}
\ifCLASSOPTIONcompsoc
  \usepackage[nocompress]{cite}
\else
  \usepackage{cite}
\fi
\ifCLASSINFOpdf
\else
\fi
\usepackage{amsmath}
\usepackage{fixltx2e}
\usepackage{algorithm}
\usepackage[noend]{algpseudocode}

\makeatletter
\renewcommand{\ALG@beginalgorithmic}{\scriptsize}

\makeatother

\usepackage{booktabs}
\usepackage{amsfonts, mathtools}
\usepackage{amsthm}
\usepackage{amsmath}
\usepackage[lambda ,advantage, operators, sets , adversary, landau, probability, notions, logic, ff , mm, primitives, events, complexity, asymptotics, keys]{cryptocode}
\usepackage{xurl}
\usepackage{multirow}
\usepackage{multicol}
\usepackage{graphicx}
\usepackage{rotating}
\usepackage{colortbl}
\usepackage{makecell}
\usepackage{tikz}
\usepackage{caption}
\usepackage{subcaption}
\usepackage{diagbox}
\usepackage[first=0,last=9]{lcg}
\usepackage{stfloats}
\usepackage{diagbox}

\usepackage{pgfplots}
\pgfplotsset{compat=1.15}

\newtheorem{theorem}{Theorem}
\newtheorem{definition}{Definition}
\newtheorem{proposition}{Proposition}
\newtheorem{corollary}{Corollary}
\newtheorem{lemma}{Lemma}

\newcommand{\round}[1]{\lfloor #1 \rceil}
\newcommand{\db}{\pckeystyle{db}}
\newcommand{\qu}{\pckeystyle{qu}}
\newcommand{\ck}{\pckeystyle{ck}}
\newcommand{\eck}{\pckeystyle{eck}}
\newcommand{\pck}{\pckeystyle{pck}}
\newcommand{\ct}{\pckeystyle{ct}}
\newcommand{\ans}{\pckeystyle{ans}}
\newcommand{\pirstate}{\pckeystyle{st}}
\newcommand{\hint}{\textbf{H}}

\newcommand{\protocol}{ZipPIR}
\newcommand{\protocolsingle}{ZipPIR$_{C}$}
\newcommand{\protocolbatched}{ZipPIR$_{B}$}
\newcommand{\addkey}{\sk_{A}}
\newcommand{\paillierkey}{\sk_{P}}
\newcommand{\paillierpk}{\pk_{P}}
\newcommand{\lwekey}{\sk}

\newcommand{\mask}{\textbf{A}}
\newcommand{\eudist}[1]{\lVert #1 \rVert}

\newcommand{\appendixsectionname}{Appendix}
\newcommand{\appsection}[1]{\appendixsectionname~\ref{#1}}

\usepackage{stackengine}
\stackMath

\usepackage{cleveref}

\begin{document}
\title{HE is all you need: Smaller FHE Responses via Additive HE}
%
%

\author{
  \IEEEauthorblockN{
    Rasoul Akhavan Mahdavi\IEEEauthorrefmark{1}, 
    Abdulrahman Diaa\IEEEauthorrefmark{1}, 
    Florian Kerschbaum\IEEEauthorrefmark{1}
  }
  \IEEEauthorblockA{\IEEEauthorrefmark{1}University of Waterloo}
}


%

\maketitle

\begin{abstract}
Homomorphic Encryption (HE) is a commonly used tool for building privacy-preserving applications.
However, in scenarios with many clients and high-latency networks, communication costs due to large ciphertext sizes are the bottleneck.
In this paper, we present a new compression technique that uses an additive homomorphic encryption scheme with small ciphertexts to compress large homomorphic ciphertexts based on Learning with Errors (LWE).
Our technique exploits the linear step in the decryption of such ciphertexts to delegate part of the decryption to the server.
We achieve compression ratios up to 90\% which only requires a small compression key.
By compressing multiple ciphertexts simultaneously, we can over 99\% compression rate.
Our compression technique can be readily applied to applications which transmit LWE ciphertexts from the server to the client as the response to a query.
Furthermore, we apply our technique to private information retrieval (PIR) where a client accesses a database without revealing its query.
Using our compression technique, we propose \protocol{}, a PIR protocol which achieves the lowest overall communication cost among all protocols in the literature.
\protocol{} does not require any communication with the client in the preprocessing phase, making it a great solution for use cases of PIR with ephemeral clients or high-latency networks.

\end{abstract}


\section{Introduction}
\label{sec:introduction}

Privacy-preserving services such as compromised credential checks and private contact discovery~\cite{liProtocolsCheckingCompromised2019a, Thomas2019ProtectingAF} have become more prevalent in recent years.
However, privacy-preserving services are associated with a computation and communication overhead, which is a barrier to using these services in many cases.
If private services can be offered with less overhead, it will make them more amenable and users more likely to access and use them.
While these services are built using a variety of tools and techniques, one commonly used tool is homomorphic encryption.

Fully Homomorphic Encryption (FHE), often deemed the holy grail of encryption, is a form of encryption which permits computation on data in encrypted form, without the need to decrypt.
Initial homomorphic cryptosystems were impractically slow~\cite{gentryFullyHomomorphicEncryption2011,vandijkFullyHomomorphicEncryption2010}.
Over the last decade, FHE has advanced significantly and there have been many advances in functionality and many orders of magnitude improvements in runtime~\cite{brakerskiFullyHomomorphicEncryption2012,fan2012somewhat,brakerskiLeveledFullyHomomorphic2012,cheonHomomorphicEncryptionArithmetic2017a,chillottiTFHEFastFully2020}.

However, another important aspect of performance for FHE is its communication costs.
FHE is frequently used to design a protocol in a client-server model~\cite{liProtocolsCheckingCompromised2019a}, where the client uploads its encrypted inputs to the server.
The server processes the information and returns the encrypted result to the client.
While such protocols require one round trip, the bandwidth usage can be extremely high because homomorphic ciphertexts are much larger than their underlying plaintext.
The server can improve the runtime via better hardware and parallelization, but it is impractical to assume all clients have a reliable and high-bandwidth network connection.

The main reason for high communication costs in FHE protocols is large ciphertexts.
Homomorphic encryption schemes based on lattices, which are the most widely used today, have large ciphertext sizes~\cite{regevLatticesLearningErrors2009,lyubashevskyIdealLatticesLearning2010}.
In the aforementioned client-server setup, communication costs consist of uploaded ciphertexts, in the form of a query, and downloaded ciphertexts which are the response.
While there are effective methods to reduce upload costs~\cite{aliCommunicationComputationTradeoffs2021a,choTranscipheringFrameworkApproximate2021a,dobraunigPastaCaseHybrid2021,dobraunigRastaCipherLow2018,albrechtCiphersMPCFHE2015,canteautStreamCiphersPractical2018}, the same techniques can not be applied to reduce download costs and a different approach is required

\emph{Approaches to Reduce Download Costs.}
Works as early as the work of Naehrig et al.~\cite{naehrigCanHomomorphicEncryption2011a} suggested switching to smaller parameters in the LWE-based schemes before sending the ciphertexts back to the client.
In RLWE-based schemes, it is common to perform a modulus switching to the lowest prime modulus before sending the results back to the client~\cite{MicrosoftSEALRelease2023,cheonHomomorphicEncryptionArithmetic2017a}.
Such techniques result in ciphertexts that do not support further homomorphic operations but contain the same message as the original ciphertext.
Despite these improvements, the size of the ciphertext is still proportionate to $O(n\log q)$ in the case of LWE, where $n$ and $q$ are the ciphertext dimension and ciphertext modulus, respectively.
Other works propose compression of GSW ciphertexts~\cite{gentryCompressibleFHEApplications2019,brakerskiLeveragingLinearDecryption2019} and LTV ciphertexts~\cite{huImprovingEfficiencyHomomorphic2013}, but these ciphertexts are not usually sent over the network and are used internally in procedures such as bootstrapping~\cite{chillottiTFHEFastFully2020}.
Moreover, size reduction is only achieved when there are many ciphertexts to compress~\cite{gentryCompressibleFHEApplications2019, brakerskiLeveragingLinearDecryption2019}.

\emph{Our Approach.}
To achieve smaller response sizes in practical applications, we propose a technique to compress ciphertexts based on LWE and RLWE, which are commonly transmitted over the network.
Our technique exploits the linear step in LWE and RLWE decryption to delegate part of the decryption to the server.
Our idea is to use additively homomorphic encryption for this step which has a much smaller ciphertext expansion factor.
However, a naive application of this idea is neither communication nor computation efficient and we develop several optimizations to help this technique advance over the state-of-the-art.
Our evaluation shows that this results in size reduction from 6.8KB to 512 bytes (about 90\%) for a single LWE ciphertext with common parameters.
By compressing many ciphertexts at once, we can achieve higher compression ratios of up to 99\%.
The only necessary information for this size reduction is a small compression key which can be reused across interactions and to compress all ciphertexts encrypted under a specific public key.
Our techniques can be readily applied to applications which transmit LWE ciphertexts back to the client in a plug-and-play fashion.
Aside from the generic application of our technique, we also identified specific applications which benefit from compression.
One such application is private information retrieval.

\emph{Problem of PIR protocols.}
In a Private Information Retrieval (PIR) protocol, the client wishes to retrieve an element from a database such that the server learns nothing about the client's desired query.
In many real-world applications of PIR, the client may not have an established, permanent connection with the server.
Many existing PIR protocols are not suitable for ephemeral, temporary clients.
More specifically, they suffer from one of two limitations:
1) They require clients to send large cryptographic keys, which the server has to store.
This is a per-client storage on the server which would only be beneficial if the clients were making many queries.
2) The server gives a database-dependant hint to the client(s).
While this hint makes online queries very fast, it is a high upfront communication cost that needs to be amortized over many queries.

\emph{\protocol{}: Low-Communication PIR using Compression}
Using our compression technique, we devise a practical PIR protocol with preprocessing, dubbed \protocol{}, which does not require any hint to be transmitted to the clients but only our small compression keys from the client to the server. 
These compression keys are small and change between interactions.
Therefore, there is no need for per-client storage on the server.
Our construction only requires client-independent preprocessing by the server, which can easily be updated as the database changes.
\protocol{} requires only 200 KB - 500 KB of communication, even for large databases, and has a competitive runtime with other protocols operating in the same model.
In summary, \protocol{} has the lowest communication of all state-of-the-art PIR protocols in the literature \cite{castroWhisPIRStatelessPrivate2024, liHintlessSingleServerPrivate2023}, making it suitable for applications with ephemeral clients in low latency networks and frequently updating databases.
\section{Background}
\label{sec:background}

In this section and throughout the paper, we index the i$^{th}$ element of the vector $\textbf{a}$ as $\textbf{a}[i]$.
We also define $[n]=\{0,1,\cdots,n-1\}$ and $\lfloor\cdot\rceil$ denotes rounding to the nearest integer.
$x\leftarrow D$ denotes the variable $x$ sampled from a distribution $D$ and $x\sample S$ denotes sampling $x$ uniformly from a set $S$.

\subsection{Homomorphic Encryption}
Homomorphic Encryption (HE) is a form of public-key cryptography which permits computation on messages while in encrypted form, without the need to access the secret key. Similar to other public-key cryptosystems, homomorphic ciphertexts are larger than the underlying plaintext.
The ratio between the ciphertext and plaintext is denoted as the \textit{expansion factor}.

Homomorphic encryption is typically used to construct a one-round protocol between a client and a server.
The client encrypts its private input homomorphically and sends the resulting ciphertexts to a server.
This constitutes the client's \textit{request}. The server computes the desired function over the client's encrypted input. The result is then transmitted back to the client as the \textit{response}. 
This work addresses the large response size and aim to reduce it.

\subsection{LWE \& RLWE ciphertexts}
\label{sec:lwe}

\begin{algorithm}[t]
     \caption{Encryption and Decryption of $\mathcal{E}_{LWE}$}
     \label{alg:lwe-encrypt-decrypt}
     \begin{algorithmic}[1]
        \vspace{1mm}
        \Procedure{LWEEncrypt}{$\texttt{sk}, \mu$}
        \State Sample $\textbf{a}\xleftarrow{\$} \ZZ_q^{n}$ and $e \leftarrow \chi$
        \vspace{1mm}
        \State $b = \sum_{i\in[n]} \textbf{a}[i] \cdot \texttt{sk}[i] + \Delta \cdot \mu + e \mod q$
        \vspace{2mm}
        \State \Return $\ct=(\textbf{a},b)$
        \EndProcedure
        \vspace{3mm}
        \Procedure{LWEDecrypt}{$\texttt{sk}, \ct=(\textbf{a},b)$}
        \State $\mu^* = \left(b - \sum_{i\in[n]} \textbf{a}[i] \cdot \texttt{sk}[i]\right) \mod q $
        \vspace{1mm}
        \State $\mu'=\lfloor\mu^* / \Delta\rceil$
        \vspace{2mm}
        \State \Return $\mu'$
        \EndProcedure
        \vspace{3mm}        
        \Procedure{RLWEEncrypt}{$S(X), \mu(X)$}
        
        \State Sample $A(X)\xleftarrow{\$} R_q$ and $E(X) \leftarrow \chi$
        \vspace{2mm}
        \State $B(X) = A(X)\cdot S(X) + \Delta \cdot \mu(X) + E(X) \mod R_q$
        \State \Return $C=(A(X), B(X))$
        \EndProcedure

        \vspace{3mm}
        
        \Procedure{RLWEDecrypt}{$S(X), C$}
        \State $(A(X), B(X)) \leftarrow C$
        \vspace{1mm}
        \State $\mu^*(X) = (B(X) - A(X)\cdot S(X)) \mod R_q$
        \vspace{2mm}
        \State $\mu'(X)=\lfloor\mu^*(X)/\Delta\rceil$
        \vspace{2mm}
        
        \State \Return $\mu'(X)\in R_p$
        \EndProcedure     
     \end{algorithmic}
\end{algorithm}

For this work, we describe a simple version of an encryption system based on the Learning With Errors (LWE)~\cite{regevLatticesLearningErrors2009} assumption which we will denote by $\mathcal{E}_{\text{LWE}}$. The most prominent encryption schemes that have ciphertext of this format are Regev~\cite{regevLatticesLearningErrors2009}, FHEW~\cite{ducasFHEWBootstrappingHomomorphic2015}, and TFHE(CGGI)~\cite{chillottiTFHEFastFully2020}.

$\mathcal{E}_{\text{LWE}}$ uses the following parameters:
dimension $n$, ciphertext modulus $q$, plaintext modulus $p$, $\Delta=\round{q/p}$, a discrete error distribution over $\ZZ_q$ called $\chi$. We sample the secret key, $\texttt{sk}$, from $\ZZ_q^{n}$. The encryption and decryption procedure for $\mathcal{E}_{LWE}$ is shown in \Cref{alg:lwe-encrypt-decrypt}.

Fresh ciphertexts can be compressed to reduce network costs. Since $\textbf{a}$ is sampled at random, we can send the seed used to generate $\textbf{a}$ instead of the vector itself. Concretely, instead of sending $c=(\textbf{a},b)$, the client can produce $\bar{c}=(\theta,b)$ where $\theta\leftarrow\{0,1\}^{\lambda}$ is the seed of a cryptographically secure PRG used to generate $\textbf{a}$, i.e., $\textbf{a}\leftarrow \texttt{PRG}(\theta)$. With this technique, fresh ciphertexts are only $\lambda+\log_2 q$ bits instead of $n\log_2 q$.

Similar to LWE, we can also construct an encryption scheme based on the Ring Learning with Errors (RLWE)~\cite{lyubashevskyIdealLatticesLearning2010} assumption, which we will denote as $\mathcal{E}_{\text{RLWE}}$. Cryptosystems such as BGV~\cite{brakerskiLeveledFullyHomomorphic2012}, BFV~\cite{brakerskiFullyHomomorphicEncryption2012,fan2012somewhat}, and CKKS~\cite{cheonHomomorphicEncryptionArithmetic2017a} have ciphertexts of a similar format.
RLWE ciphertexts are useful since they can encrypt a polynomial, i.e. a vector of numbers, instead of just one scalar. For RLWE encryption, we select a dimension $N$, ciphertext modulus $q$, plaintext modulus $p$, and $\Delta=\round{q/p}$. Define $R_q=\ZZ_q[X]/(X^N+1)$ and $R_p$ similarly. Moreover, define a discrete error distribution $\chi$ over $R_q$.
For key generation, sample $S(X)$ uniformly from $R_q$.

Similar to LWE, we can also compress fresh RLWE ciphertexts by sending the seed used to generate $A(X)$~\cite{aliCommunicationComputationTradeoffs2021a, MicrosoftSEALRelease2023, albadawiOpenFHEOpenSourceFully2022}. Using this technique, the size of a ciphertext can be reduced from $2N\log_2 q$ bits to $\lambda + N\log_2 q$.

\subsection{Private Information Retrieval}
Private Information Retrieval (PIR) is a protocol where a client retrieves an element from a database such that the query is not revealed.
A specific variant of PIR is PIR-with-preprocessing which consists of four routines.

\begin{itemize}
    \item $\hint \leftarrow \textsc{Setup}(\db)$ : Create database-dependant hint
    \item $(\pirstate, \qu) \leftarrow \textsc{Query}(i)$ : Generates the client query
    \item $\ans \leftarrow \textsc{Response}(\db, \hint, \qu)$ : Computes response
    \item $d \leftarrow \textsc{Extract}(\pirstate, \ans)$ : Extracts result from response
\end{itemize}

PIR-with-preprocessing allows the server to perform most of the preprocessing offline such that the online stage is very fast. 
Note that our definition is slightly relaxed compared to previous definitions ~\cite{henzingerOneServerPrice2023, beimelReducingServersComputation2000} as it does not generate a client-specific hint.

\begin{definition}[Correctness]
    A PIR protocol with preprocessing consisting of the four aforementioned routines is $\delta$-correct, if for a domain $\mathcal{D}$, any database $\db\in \mathcal{D}^{N}$ and any $i\in[N]$, 
    \begin{align}
    \PP\left[\db[i] = f \middle|
        \begin{array}{c}
            \hint \leftarrow \textsc{Setup}(\db) \\
            (\st, \qu) \leftarrow \textsc{Query}(i) \\
            \ans \leftarrow \textsc{Respond}(\db, \hint, \qu) \\
            f \leftarrow \textsc{Extract}(\st, \ans)
        \end{array}
    \right] > 1-\delta
    \end{align}.
\end{definition}

Intuitively, the query should not reveal any information about the record that it is querying. This is formalized in the following definition.
\begin{definition}[Security]
     A PIR protocol is $\epsilon$-secure if for any PPT adversary $\adv$ and any $i,j\in[N]$,
    \begin{align}
        | \PP[\adv&(1^N, \qu) = 1 | (\st, \qu) \leftarrow \textsc{Query}(i) ] \\
        & - \PP[\adv(1^N, \qu) = 1 | (\st, \qu) \leftarrow \textsc{Query}(j) ] | \leq \epsilon
    \end{align}
\end{definition}
\section{Additive HE for Smaller FHE Responses}
\label{sec:main}

Ciphertexts that have been processed by the server can not be compressed using the technique mentioned in \Cref{sec:background}. We propose a technique to compress LWE/RLWE ciphertexts using auxiliary information provided by the client.

\emph{Exploiting Linear Phase Evaluation.}
In LWE and RLWE decryption, we compute an intermediate value which is commonly referred to as the \textit{phase}, i.e., $\mu^*$ and $\mu^*(X)$ in \Cref{alg:lwe-encrypt-decrypt}.
Phase evaluation is linear in both schemes and the phase is much smaller than the ciphertext itself.
The main insight behind our solution is to homomorphically compute the phase on the server using encrypted values of the secret key, encrypted under an additive encryption scheme.
Since the phase is much smaller than the original ciphertext, this results in a smaller response size.
In general, our technique can be applied to any encryption scheme that has a linear phase evaluation. Examples of encryption schemes with this property are Regev~\cite{regevLatticesLearningErrors2009}, FHEW~\cite{ducasFHEWBootstrappingHomomorphic2015}, TFHE~\cite{chillottiTFHEFastFully2020}, BFV~\cite{fan2012somewhat,brakerskiFullyHomomorphicEncryption2012}, and BGV~\cite{brakerskiLeveledFullyHomomorphic2012}.


\emph{The Additive Encryption Scheme.}
For the compression protocol, we require an additive encryption scheme which we denote $\mathcal{E}_A$ such that the plaintext space is $\ZZ_m$, for some $m$.
Also, denote the ciphertext space of $\mathcal{E}_A$ as $\mathcal{C}$.
$\mathcal{E}_A$ supports addition and plaintext multiplications.
We denote addition and plaintext multiplication with $\oplus$ and $\otimes$, respectively.
Moreover, denote the secret key generated by $\mathcal{E}_A$ as $\addkey$ and the corresponding encryption and decryption algorithms as $\texttt{AEnc}$ and $\texttt{ADec}$.

Paillier~\cite{paillierPublicKeyCryptosystemsBased1999}, Damgard-Jurik~\cite{damgardGeneralisationSimplificationApplications2001a}, Exponential ElGamal~\cite{elgamalPublicKeyCryptosystem1985}, and Benaloh~\cite{benalohDenseProbabilisticEncryption1994} are examples of cryptosystems that can be used for this purpose.

\subsection{Compressing LWE Ciphertexts}

The ciphertext compression algorithm for LWE and the corresponding modified decryption algorithm is given in \Cref{alg:lwe-compress}.

\begin{algorithm}[H]
    \caption{LWE compression, performed by the server and the corresponding modified decryption process, performed by the client over a compressed ciphertext. The compression key $\ck\in\mathcal{C}^{n}$ is such that $\ck[i]=\texttt{AEnc}(\addkey, \sk[i])$.}
    \label{alg:lwe-compress}
    \begin{algorithmic}[1]
        \Procedure{LWECompress$_{q}$}{$\ck, \ct=(\textbf{a},b)$}
        \Comment{$\ct\in\ZZ^{n}\times\ZZ$}
        \State $x=b$
        \For{$i \in [n]$} 
            \State $x \leftarrow x \oplus \left( (q-\textbf{a}[i]) \otimes \ck[i] \right)$\label{alg:line-multiply}
        \EndFor
        \State \Return $x$ \Comment{$\mu^*=\texttt{ADec}(\addkey, x)$}
        \EndProcedure
    \vspace{3mm}
    \Procedure{ModifiedLWEDecrypt$_{q, p}$}{$\addkey,x$}
	 	\State $\mu^{**} = \texttt{ADec}(\addkey, x) \mod q$
            \label{alg:lwe-mod-decrypt}
	 	\vspace{1mm}
	 	\State $ \mu'' = \lfloor \mu^{**}/\Delta\rceil$
            \Comment{$\Delta=\round{q/p}$}
            \vspace{2mm}
            \State \Return $\mu'' \in \ZZ_{p}$
   \EndProcedure
   \end{algorithmic}
\end{algorithm}

\begin{theorem}[Correctness]
\label{thm:lwe-compress-correct}
    For an LWE ciphertext $\ct\in\ZZ_q^{n+1}$, if $m>q+nq^2$, then $\textsc{LWECompress}_q$ produces a compressed ciphertext which decrypts to the correct message if decrypted using \textsc{ModifiedLWEDecrypt}. More formally, if 

\begin{align*}
    x\leftarrow\textsc{LWECompress}_{q}(\ck, \ct) \\
    \mu'' \leftarrow \textsc{ModifiedLWEDecrypt}_{q,p}(\addkey, x)
\end{align*}
then
$\mu'' = \textsc{LWEDecrypt}(\sk, \ct)$
\end{theorem}

\begin{proof}
In the \textsc{ModifiedLWEDecrypt}$_{q,p}$ procedure (\Cref{alg:lwe-mod-decrypt} of \Cref{alg:lwe-compress}), we calculate 
$b + \sum_{i\in[n]} (q-\textbf{a}[i]) \cdot \sk[i]$, encrypted under additive encryption, which is achievable due to the linear properties of the additive encryption. We know that $\sk[i], \textbf{a}[i]$ and $b$ are elements in $\ZZ_q$ so $0 \leq \sk[i], \textbf{a}[i], b < q$ and 

{\footnotesize
\begin{align}
    b + \sum_{i\in[n]} (q-\textbf{a}[i]) \cdot \sk[i] \leq q + \sum_{i\in[n]} q  \cdot q = q + nq^2 < m .
\end{align}
}

so there is no overflow in the plaintext space of the additive ciphertext.
In $\textsc{ModifiedLWEDecrypt}_{q,p}$ (\Cref{alg:lwe-compress}), we have

{\footnotesize
\begin{align*}
    \mu^{**}\mod q &= \texttt{ADec}(\addkey, x) \mod q \\
    &= \left((b + \sum_{i\in[n]} (q-\textbf{a}[i]) \cdot \sk[i]) \mod m \right) \mod q \\
    &= \left(b + \sum_{i\in[n]} (q-\textbf{a}[i]) \cdot \sk[i] \right) \mod q \\
    &= b - \sum_{i\in[n]} \textbf{a}[i] \cdot \sk[i] \mod q
\end{align*}
}

This is identical to $\mu^*$ in line 1 of \Cref{alg:lwe-encrypt-decrypt}, hence, since the subsequent steps of \textsc{LWEDecrypt} and \textsc{ModifiedLWEDecrypt} are identical, they produce the same response, and the theorem is proven.
\end{proof}

In cryptosystems such as TFHE~\cite{chillottiTFHEFastFully2020}, the secret key is sampled from a binary distribution.
In such a case, we can tighten the inequality required in \Cref{thm:lwe-compress-correct} for correctness because $0\leq\sk[i]\leq 1$. The following corollary summarizes this fact.

\begin{corollary}
    If the LWE secret key is binary and $m>q+nq$, \textsc{LWECompress} produces a compressed ciphertext which decrypts to the correct message if decrypted using \textsc{ModifiedLWEDecrypt}.
\end{corollary}

In Gentry's original construction of a bounded depth encryption scheme, he proposed the idea of using a chain of semantically secure cryptosystems, such that each cryptosystem encrypts the secret key of the next~\cite{gentryFullyHomomorphicEncryption2009}. Gentry proved that if the secret key of each cryptosystem is sampled independently, the composed scheme is also semantically secure.

Let $\mathcal{E}'$ denote the cryptosystem which is the chaining of $\mathcal{E}_{\text{LWE}}$ and $\mathcal{E}_{A}$. The encryption and decryption procedure of $\mathcal{E}'$ is shown in \Cref{alg:lwe-encrypt-decrypt} and \Cref{alg:lwe-compress}, respectively. The secret key of $\mathcal{E}'$ is the combination of the secret keys of $\mathcal{E}_{\text{LWE}}$ and $\mathcal{E}_{A}$. The same holds for the public key as well. Moreover, we also release encryptions of the bits of the secret key of $\mathcal{E}_{\text{LWE}}$ under the secret key of $\mathcal{E}_{A}$. 

\begin{proposition}[Security]
    If $\mathcal{E}_{\text{LWE}}$ and $\mathcal{E}_A$ are semantically secure, then $\mathcal{E}'$ is also semantically secure.
\end{proposition}

\subsection{Using Smaller Compression Keys}
\label{sec:smaller-compression-key}
In practice, the plaintext space of the additive encryption system could be much larger than is required for the correctness of the compression technique to hold, i.e., $m \gg q + n q^2$.
For example, the plaintext space of Paillier for 128-bit security is 3072 bits, which is much larger than $q + n q^2$ for any common choice of LWE parameters.
We can use this gap to pack multiple bits of the LWE secret key within one additive ciphertext.
Instead of encrypting each bit of the LWE key separately, we encrypt the first $t$ bits of the secret key together into one packed additive ciphertext as $\texttt{pck}_{0-t} = \texttt{AEnc}(\addkey, \sum_{i\in[t]}\sk[i] \cdot \delta^{i})$ for a large enough $\delta$.
Specifically, $\delta$ should be such that $\delta > q + n q^2$ (or $\delta > q + n q$ in the case of binary keys).
On the server side, the server unpacks the secret key by computing $\ck[i] = \delta^{t-1-i} \otimes \texttt{pck}_{0-t}$ for $i\in[t]$.
Compression proceeds as before, with the only difference being that the encrypted phase, calculated by the server in the additive ciphertext, is scaled by a factor of $\delta^{t-1}$.
\appsection{sec:lwe-compress-packed-keys} details the procedures for generating the packed key, unpacking it, and the corresponding modified LWE decryption function.
We use the same function for compressing the ciphertext.


\subsection{Batched Compression}
\label{sec:batched-compression}

To achieve better compression, multiple LWE ciphertexts (encrypted using the same secret key) can be compressed within the same additive ciphertext, which we denote as \textit{batched compression}.
Each LWE ciphertext takes up $\log_2 (q+nq^2)$ bits of the total bitwidth of the plaintext space.
So, if $m$ is the modulus of the plaintext space, then $\floor{\log_2 m / \log_2 (q+nq^2)}$ LWE ciphertexts can be compressed into one ciphertext from the additive cryptosystem.

\Cref{alg:batched-compress} illustrates how to compress $\ell$ LWE ciphertexts within one additive ciphertexts. The corresponding decryption procedure is also shown.
Using \textsc{LWECompress} as a subprocedure allows for better parallelization when compressing many LWE ciphertexts.

\begin{algorithm}[t]
	 \caption{Batch Compression of LWE ciphertexts by the server and the modified decryption procedure, performed by the client. The compression key $\ck$ is such that $\sk[i]={\texttt{ADec}(\addkey, \ck[i])}$ and $cts=\{c_j\}_{j\in [\ell]}$ such that $c_j = (\boldsymbol{\textbf{a}_j}, b_j)\in \ZZ_{q}^{n}\times\ZZ$ and $\gamma = q + n q^2$.}
	 \label{alg:batched-compress}
	 \begin{algorithmic}[1]
    \Procedure{BatchedLWECompress$_{q, \gamma}$}{$\ck, cts=\{c_j\}_{j\in [\ell]}$}
        \For{$j\in{[\ell]}$}
            \State $x_j \leftarrow \textsc{LWECompress}_q(\ck,c_j)$
            \State $x \leftarrow x \oplus \gamma^{j} x_j $
        \EndFor
        \Return $x$
    \EndProcedure
    \vspace{3mm}
    \Procedure{ModifiedBatchedLWEDecrypt$_{q,p, \gamma}$}{$\addkey, x$}
	 	\State $\mu^{**} = \texttt{ADec}(\addkey, x)$
            \For{$j \in [\ell]$}
                \State $\mu_j^{**} =\floor{\mu^{**}/\gamma^j} \mod \gamma$
	       \State $\mu''_j = \lfloor \frac{\mu_j^{**} \mod q}{\Delta} \rceil$\Comment{$\Delta=\round{q/p}$}
            \EndFor
        \Return $\{\mu_j''\in \ZZ_{p}\}_{j\in[\ell]}$
    \EndProcedure
	 \end{algorithmic}
\end{algorithm}

\begin{theorem}[Correctness]
    Let $\ct=\{c_j\}_{j\in [\ell]}$ be a vector of $\ell$ LWE ciphertexts.
    For $\gamma\geq q + n q^2$, if  $m > \gamma^{\ell}$, then \textsc{BatchedLWECompress}$_{q, \gamma}$ produces a compressed ciphertext which, if decrypted using the corresponding modified decryption, decrypts to the vector of $\ell$ plaintexts. More formally, if
    \begin{align*}
        x \leftarrow \textsc{BatchedLWECompress}(\ck, \ct, k) \\
        \{\mu'_j\}_{j\in\ell} \leftarrow \textsc{ModifiedBatchedLWEDecrypt}(
        \addkey, x)
    \end{align*}
    then $\mu'_j = \text{LWEDecrypt}(\sk,c_j)$.
\end{theorem}
\begin{proof}
    By the proof of \Cref{thm:lwe-compress-correct} we know that if $\mu^{**}_j=\texttt{ADec}_{\texttt{s}}(x_j)$, then $0\leq \mu^{**}_j < \gamma = q + nq^2$. Hence, we have
    \begin{align*}
        \mu^{**} = \sum_{j\in[\ell]} \gamma^j \mu^{**}_j \leq \sum_{j\in[\ell]} \gamma^j (\gamma-1) = \gamma^{\ell} - 1 < \gamma^{\ell} < m .
    \end{align*}
    Hence, the plaintext corresponding to $x$, i.e., $\mu^{**}$, does not overflow in the plaintext space of the additive ciphertext.
    If $\mu^*_j$ us equivalent to $\mu^*$ in the \textsc{LWEEncrypt} procedure, then for some value $t$,
    \begin{align*}
        \mu''_j =\floor{\mu^{**}/\gamma^j} \mod \gamma = (\mu^*_j + \gamma \cdot t ) \mod \gamma = \mu^*_j
    \end{align*}
    and the subsequent steps are similar, which proves the theorem.
\end{proof}

\subsubsection{Faster Batched Compression with Expanded Key}
Compression makes use of expensive operations in the additive scheme.
The plaintext multiplication in \Cref{alg:line-multiply} of \Cref{alg:lwe-compress} is the most expensive operation.
In additive schemes such as Paillier and ElGamal, this is equivalent to a modular exponentiation in a large group.

In the batched setting, we can reduce the overhead by precomputing and reusing multiples of the bits of the secret key. If we decompose $(q-\textbf{a}[i])$ as $(b_{t-1}\cdots b_1 b_0)_2 = (q-\textbf{a}[i]) \mod q$ we compute the plaintext multiplication as follows
\begin{align}
    &(q-\textbf{a}[i]) \otimes \ck[i] \\
    &= 2^{t-1} b_{t-1} \ck[i] + \cdots + 2b_{1} \ck[i] + b_{0} \ck[i]
\end{align}
and we can precompute and \textit{extended compression key}, $\eck$, such that $\eck[i][j] = 2^j \ck[i]$ for $j \in [t]$, which can reused for all LWE ciphertexts we want to compress.

\begin{algorithm}[t]
	 \caption{Batch compression of LWE ciphertexts using precomputed powers. The compression key $\ck$ is such that $\sk[i]={\texttt{ADec}(\addkey, \ck[i])}$ and $cts=\{c_j\}_{j\in [\ell]}$ such that $c_j = (\boldsymbol{\textbf{a}_j}, b_j)\in \ZZ_{q}^{n}\times\ZZ$.}
	 \label{alg:batched-compress-precompute}
	 \begin{algorithmic}[1]
    \Procedure{ExpandCompressionKey$_{q}$}{$\ck$}
        \State $\eck[0] = \ck$
        \For {$i \in [t-1]$} \Comment{$t = \ceil{\log_2 q}$}
            \For {$j \in [n]$}
                \State $\eck[i+1][j] = \eck[i][j] \oplus \eck[i][j]$
            \EndFor
        \EndFor
        \Return $\eck$
    \EndProcedure
    \vspace{3mm}
    \Procedure{FastLWECompress$_q$}{$\eck, \ct=(\textbf{a},b)$}
    \State $x=b$
    \For{$i \in [n]$} 
        \State $(b_{t-1}\cdots b_1 b_0)_2 \leftarrow (q-\textbf{a}[i]) \mod q$
        \Comment{$t = \ceil{\log_2 q}$}
        \For {$j \in [t]$}
            \If {$b_j = 1$}
                \State $x \leftarrow x \oplus \eck[j][i]$
            \EndIf
        \EndFor
    \EndFor
    \Return $x$ \Comment{$\mu^*=\texttt{ADec}(\addkey, x)$}
    \EndProcedure
    \vspace{3mm}
    \Procedure{FastBatchedLWECompress$_{q,\gamma}$}{$\ck, cts=\{c_j\}_{j\in [\ell]}$}
        \State $\eck \leftarrow \textsc{ExpandCompressionKey}_{q}(\ck)$
        \State $\gamma = q + n q^2$
        \For{$j\in{[\ell]}$}
            \State $x_j \leftarrow \textsc{FastLWECompress}_q(\eck,c_j)$
            \State $x \leftarrow x \oplus \gamma^{j} x_j $
        \EndFor
        \Return $x$
    \EndProcedure
	 \end{algorithmic}
\end{algorithm}

\begin{corollary}
    Let $\ct=\{c_j\}_{j\in [\ell]}$ be a vector of $\ell$ LWE ciphertexts.
    For $\gamma\geq q + n q^2$, if  $m > \gamma^{\ell}$, if
    \begin{align*}
        x \leftarrow \textsc{FastBatchedLWECompress}(\eck, \ct, k) \\
        \{\mu'_j\}_{j\in\ell} \leftarrow \textsc{ModifiedBatchedLWEDecrypt}(
        \addkey, x)
    \end{align*}
    then $\mu'_j = \text{LWEDecrypt}(\sk,c_j)$.
\end{corollary}

\subsubsection{Rescaling for Compression}

In some instances, it is possible to rescale the elements in the ciphertext to a smaller modulus without altering the underlying message.
This technique, also called modulus switching, is commonly used in the literature to simplify the decryption procedure or control noise growth~\cite{brakerskiFullyHomomorphicEncryption2012}.
However, rescaling is only possible if the noise of the underlying LWE ciphertext is less than a given bound. 
In \appsection{sec:modulus-switching-theorem}, we prove how rescaling is possible for LWE ciphertexts with binary keys, if the noise is less than a certain bound, i.e., less than $\Delta/4$.
Rescaling to a smaller modulus accelerates our compression technique since the scalar multiplication in the additive encryption scheme is done with a smaller scalar.

\subsubsection{Better compression with a smaller scale}
\label{sec:overlapping-noise}

The number of LWE ciphertexts that fit within each additive ciphertext is determined by the scale, i.e., $\gamma=q+nq^2$.
Using a smaller scale would allow us to pack more LWE ciphertexts within each additive ciphertext.
There are two instances where we can use a smaller scale.
First, when the LWE secret key is binary.
In that case, we can use $\gamma=q+nq$ as the scale.
This follows from the fact that in the case of binary keys, $0<\mu_j^{**} \leq \gamma = q+nq^2$.

The second instance where we can reduce the scale is by allowing $\mu_{j}^{**}$ and $\mu_{j+1}^{**}$ to overlap in the additive scheme.
Intuitively, this is possible because the high-order bits of $\mu_{j}^{**}$ are removed when it is modulized by $q$ as part of the modified decryption.
The lower order bits of $\mu_{j+1}^{**}$ are also rounded during the modified decryption so it is possible to add additional error, as long as it does not interfere with the message.
Specifically, if $|e|<\Delta/4$ (instead of the usual condition where $|e|<\Delta/2$ for correct decryption), we can reduce the scale to $\gamma=q^2$ and $\gamma=q$ in the case of non-binary and binary keys, respectively.
Due to space restrictions, we provide proof of the correctness of this technique using a smaller scale under these conditions in the full version of the paper.

\subsection{Compressing RLWE Coefficients}
\label{sec:rlwe-compression}

RLWE ciphertexts also have a linear phase evaluation and hence, can benefit from our compression technique.
However, an RLWE ciphertext is only twice as large as the phase so the compression technique, applied naively, would not yield a significant improvement.
Our approach is beneficial if the user is only interested in some coefficients of the RLWE plaintext and not all of them.

The main observation is that each coefficient of $\mu'(X)$ in the \textsc{RLWEDecrypt} procedure can be calculated separately. Specifically, for $0\leq k \leq N-1$

\makeatletter
    \def\tagform@#1{\maketag@@@{\normalsize(#1)\@@italiccorr}}
\makeatother

{\tiny
\begin{align}
    \mu'&[k] = \lfloor\frac{\mu^*[k]}{\Delta}\rceil \\
    &= \left\lfloor \frac{B[k] - \sum_{i=0}^{k} A[k-i] \cdot S[i] + \sum_{i=k+1}^{N-1} A[N+k-i] \cdot S[i]}{\Delta} \right\rceil
    \label{eq:rlwe-extract-general}
\end{align}
}%

Note that the operations in the numerator are happening modulo $q$. The numerator of \Cref{eq:rlwe-extract-general} is a linear combination of the coefficients of the secret key, hence it can be computed given the encrypted coefficients of the secret key. The complete procedure to compress the coefficient of $X^k$ in an RLWE plaintext and the corresponding decryption function is shown in \Cref{alg:rlwe-compress-response}. Compression of RLWE coefficients is fully compatible with the compact compression keys of \Cref{sec:smaller-compression-key} and batched compression of \Cref{sec:batched-compression}.

\begin{algorithm}[H]
    \caption{
      Compressing Extracted RLWE Coefficient, performed by the server and the corresponding modified decryption process, for the client.
      The compression key is $\ck$ such that $\ck[i]={\texttt{AEnc}_{\texttt{s}}(S[i])}$ and $C\in R_q\times R_q$
    }
	 \label{alg:rlwe-compress-response}
	 \begin{algorithmic}[1]
    \Procedure{RLWECompressCoefficient}{$\ck, C, k$}
        \State $x=B[k]$
        \For{$i \in \{0,1,\cdots,k\}$}
            \State $x \leftarrow x \oplus \left( (q-A[k-i]) \otimes \ck[i]\right)$
        \EndFor
        \For{$i \in \{k+1,\cdots,N-1\}$}
            \State $x \leftarrow x \oplus \left(A[N+k-i] \otimes \ck[i]\right)$
        \EndFor
        
        \State \Return $x$ 
    \EndProcedure
    \vspace{3mm}
    \Procedure{ModifiedRLWEDecrypt$_{q,p}$}{$\addkey, x$}
        \State $\mu^{**}_k = \texttt{ADec}(\addkey, x) \mod q$
        \vspace{1mm}
        \State $\mu''_k= \lfloor \mu^{**}_k / \Delta \rceil$
        \Comment{$\Delta=\round{q/p}$}
        \vspace{2mm}
        
        \State \Return  $\mu''_k \in \ZZ_{p}$
    \EndProcedure
	 \end{algorithmic}
\end{algorithm}

\begin{theorem}[Correctness]
\label{thm:rlwe-compress-correct}
    If $m > q + N q^2$, \Cref{alg:rlwe-compress-response} produces a compressed ciphertext which decrypts to the coefficient of $X^k$ if decrypted using \textsc{ModifiedRLWEDecrypt}$_{q,p}$. More formally, 
    \begin{align*}\normalfont
        x \leftarrow\textsc{RLWECompressCoefficient}(\ck, c, k) \\
        \mu_k'' \leftarrow\textsc{ModifiedRLWEDecrypt}_{q,p}(\texttt{s}, x)
    \end{align*}
    then $\mu_k''$ is equal to the coefficient of $X^k$ in 
    \begin{align*}
        \mu'(X) = \textsc{RLWEDecrypt}(\sk, c)
    \end{align*}
\end{theorem}

We provide the proof of \Cref{thm:rlwe-compress-correct} in \appsection{sec:prove-rlwe-compress}.
Similar to the case of LWE ciphertexts, if the coefficients of the RLWE secret key are binary, we can tighten the condition on $m$ in \Cref{thm:rlwe-compress-correct} such that $m > q + N q$.
The following corollary summarizes this fact.

\begin{corollary}
    If the coefficients of the secret key are binary and $m > q + N q$, \Cref{alg:rlwe-compress-response} produces a compressed ciphertext which decrypts to the coefficient of $X^k$ if decrypted using \textsc{ModifiedRLWEDecrypt}$_{q,p}$.
\end{corollary}

\emph{Security.} A similar argument can be made about the security of compression over RLWE. Let $\mathcal{E}''$ denote the cryptosystem which is the combination of $\mathcal{E}_{RLWE}$ and $\mathcal{E}_{A}$. The following proposition holds regarding security.

\begin{proposition}[Security]
    If $\mathcal{E}_{RLWE}$ and $\mathcal{E}_A$ are semantically secure, then $\mathcal{E}''$ is also semantically secure.
\end{proposition}

\section{Related Work}
\label{sec:related-work}

\begin{table}[t]
    \centering
    \caption{
        Taxonomy of different techniques which involve conversion between encryption schemes.
        Techniques which could offer compression for the downloaded ciphertexts are indicated in bold.
    } 
    \resizebox{\columnwidth}{!}{%
    \begin{tabular}{cc|c}
    \toprule
    \textbf{Source} & \textbf{Dest.} & \textbf{Technique} \\
    \midrule
    \midrule
    \multirow{2}{*}{\begin{tabular}{c}Symmetric\\Ciphers\end{tabular}}
    & LWE & Hybrid HE~\cite{dobraunigPastaCaseHybrid2021}\\\cmidrule{2-3}
    & RLWE & Hybrid HE~\cite{canteautStreamCiphersPractical2018, albrechtCiphersMPCFHE2015, dobraunigRastaCipherLow2018, choTranscipheringFrameworkApproximate2021a, dobraunigPastaCaseHybrid2021} \\
    \midrule
        LTV & LWE & \textbf{Compression~\cite{huImprovingEfficiencyHomomorphic2013}}\\
    \midrule
        \multirow{4}{*}{LWE} & LWE &
        \begin{tabular}{c}
            Keyswitching~\cite{chillottiTFHEFastFully2020, ducasFHEWBootstrappingHomomorphic2015} \\
            \textbf{Dim. Reduction~\cite{brakerskiEfficientFullyHomomorphic2011,naehrigCanHomomorphicEncryption2011a}}
        \end{tabular}\\\cmidrule{2-3}
         & RLWE & Scheme Switching~\cite{bouraCHIMERACombiningRingLWEbased2020} \\\cmidrule{2-3}
         & Paillier & \textbf{Our work} \\
    \midrule
        \multirow{2}{*}{\begin{tabular}{c}Multiple\\LWEs\end{tabular}}
        & RLWE & \textbf{RLWE Packing~\cite{chenEfficientHomomorphicConversion2021}} \\\cmidrule{2-3}
        & Damgard-Jurik & \textbf{Our work} \\
    \midrule
        \multirow{5}{*}{RLWE} & LWE & Coefficient Extraction~\cite{chillottiTFHEFastFully2020,bouraCHIMERACombiningRingLWEbased2020} \\\cmidrule{2-3}
        & RLWE$^*$ & Oblivious Expand~\cite{angelPIRCompressedQueries2018a,chenOnionRingORAM2019} \\\cmidrule{2-3}
        & RLWE & \textbf{Modulus Switching~\cite{brakerskiLeveledFullyHomomorphic2012,cheonHomomorphicEncryptionArithmetic2017a}} \\\cmidrule{2-3}
        & Paillier & \textbf{Our work} \\
    \midrule
        \multirow{2}{*}{GSW} & ElGamal & \cite{gentryFullyHomomorphicEncryption2011} \\\cmidrule{2-3}
        & PVW & \textbf{\cite{gentryCompressibleFHEApplications2019}}\\
    \bottomrule
    \end{tabular}%
    }
    \label{tab:sok}
\end{table}

Our technique is predicated on converting the scheme under which the plaintexts are encoded.
This concept has been proposed in the literature before but for different purposes.
In this section, we examine the existing literature which involves scheme switching in any form and outline how our work differs from them.
We specifically point out scheme-switching techniques that result in any compression, but we emphasize that our work achieves better compression in a plug-and-play fashion, which has not been the case in previous work. \Cref{tab:sok} summarizes the related work.

\subsection{Identity Scheme Switching}
Identity scheme switching implies that the destination ciphertext encrypts the exact same message as the original plaintext.
There are many techniques in the literature for conversions, either between schemes or within a scheme, with different purposes besides compression.
Some might also come with the benefit of compression, but the compression rate is not high.
The common feature of these works is that the destination ciphertext accurately encrypts the same message as the initial ciphertext.
As we will see, in other forms of scheme scheme switching, this may not hold.

\subsubsection{Parameter Switching}
There are many ways that the parameters of the encryption system may change. Here we discuss three of such methods:
Key switching, Dimension reduction, and Modulus Switching. We also describe how each method is relevant to the compression.

Key switching is a very common technique and as the name suggests, allows the secret key under which the ciphertext is encrypted, to change.
However, \cite{brakerskiEfficientFullyHomomorphic2011} observed that, for LWE-based schemes in particular, the ciphertext dimension can also change whilst changing the key.
They referred to this as dimension reduction and used it to reduce the size of the decryption circuit for better bootstrapping.
Dimension reduction can be used for compression, by simply switching to smaller parameters, i.e., smaller $n$ and $q$ as per the notation of \Cref{sec:lwe}.
Ciphertexts encrypted with the new parameters are smaller but the compression achieved by this approach is limited given that the ciphertext is still a vector in $\ZZ_q^{n+1}$, albeit with smaller parameters.

Another relevant technique is modulus switching which is primarily used in RLWE-based schemes such as BGV~\cite{brakerskiLeveledFullyHomomorphic2012}, B/FV~\cite{brakerskiFullyHomomorphicEncryption2012,fan2012somewhat}, and CKKS~\cite{cheonHomomorphicEncryptionArithmetic2017a}.
Generally, there are two purposes to modulus switching in these schemes:
First, to limit noise growth with homomorphic operations.
This technique was first proposed by Brakerski and Vaikuntanathan~\cite{brakerskiEfficientFullyHomomorphic2011} and subsequently used by Brakerski et al. to construct the BGV cryptosystem~\cite{brakerskiLeveledFullyHomomorphic2012}.
This technique is also used in the CKKS scheme to switch between levels~\cite{cheonHomomorphicEncryptionArithmetic2017a}.

The second benefit of modulus switching is size reduction before communicating the result to the client. Using our notation from \Cref{sec:lwe}, this technique results in a smaller $q$.
However, this technique is limited by the fact that the size of the ciphertext is still linear in $N$, which significantly impacts the size of the ciphertext.

Note that the ciphertext compression technique mentioned in this work can be used in conjunction with modulus switching. After the modulus has been switched to the smallest value, our compression to an additive scheme is performed. Switching the smaller parameters via modulus switching improves the compression that we can gain from our technique because more ciphertexts can fit within one additive ciphertext.

\subsubsection{Bootstrapping \& Extended Functionality}

Gentry and Halevi~\cite{gentryFullyHomomorphicEncryption2011} use scheme switching as an alternative to squashing the decryption circuit in the bootstrapping process.
They start with ciphertexts encrypted under a Somewhat Homomorphic Encryption (SWHE) scheme.
They express the decryption function of that scheme as a depth-3 circuit of a particular form.
For one step of the decryption, which involves multiplications, they switch to the Elgamal scheme~\cite{elgamalPublicKeyCryptosystem1985} which is a multiplicative encryption scheme.
They perform the multiplications in Elgamal and then switch back to the SWHE scheme by evaluating the decryption circuit of Elgamal.
This approach to bootstrapping avoids the squashing step proposed by Gentry in the original blueprint~\cite{gentryFullyHomomorphicEncryption2009} and does not require the additional assumption that the sparse subset problem is hard.

Boura et al.~\cite{bouraCHIMERACombiningRingLWEbased2020} propose scheme switching as a method to benefit from the features of many cryptosystems.
The authors provide procedures to switch between three Ring-LWE based schemes, B/FV~\cite{fan2012somewhat,brakerskiFullyHomomorphicEncryption2012}, TFHE~\cite{chillottiTFHEFastFully2020}, and CKKS~\cite{cheonHomomorphicEncryptionArithmetic2017a}.

\paragraph{Coefficient Extraction.}
Coefficient Extraction~\cite{chillottiTFHEFastFully2020,bouraCHIMERACombiningRingLWEbased2020} is a specific instance of scheme switching which we elaborate on due to the relevance to compression. Coefficient extraction generates a LWE-based ciphertext which encrypts one coefficient of an RLWE plaintext.
Given that RLWE ciphertexts are usually larger than LWE ciphertexts, conversion from RLWE to LWE can offer compression when only one RLWE coefficient is of interest.s

\paragraph{RLWE Packing.}
The reverse process of coefficient extraction is RLWE packing, which encodes many LWE plaintexts into the coefficients of an RLWE plaintext.
Due to the smaller expansion factor of RLWE ciphertexts, this technique is suitable for compression if enough coefficients in the RLWE plaintext are utilized.
Chen et al.~\cite{chenEfficientHomomorphicConversion2021} demonstrated an efficient method to perform this conversion.

\begin{table*}
    \centering
    \caption{Evaluation of the ciphertext compression technique for a single LWE ciphertext. Three sample parameter sets are chosen for LWE-based ciphertexts.
    The first three columns are common parameter sets used in the Concrete library~\cite{zamaConcreteTFHECompiler2022}.
    The last configuration is the STD128 configuration for CGGI in OpenFHE~\cite{albadawiOpenFHEOpenSourceFully2022}.} \resizebox{\textwidth}{!}{
    \begin{tabular}{c|c|c|c|c|c|c|c|c}
    \toprule
    \multirow{2}{*}{Parameters}
    & \multicolumn{4}{c|}{LWE $(n,\log_2 q)$} 
    & \multicolumn{4}{c}{RLWE $(N,\log_2 q)$} 
    \\
        & (630, 64)
        & (742, 64)
        & (870, 64)
        & (1305, 11)
        & (1024,27)
        & (2048,54)
        & (4096,36)
        & (8192,43)
         \\
    \midrule                     
        Compression Time   & 9.7 ms & 11.0 ms & 12.9 ms & 16.6 ms & 7.2 ms & 23.8 ms & 33.8 ms & 83.3 ms \\
        Compressed Ciphertext    & 768 B   & 768 B  & 768 B & 768 B   & 768 B  & 768 B & 768 B & 768 B \\
        Uncompressed Ciphertext  & 5.05 KB & 5.94 KB & 6.97 KB & 1.80 KB & 3.46 KB & 13.83 KB & 18.44 KB & 44.04 KB \\
        Size Reduction            & \textbf{84.78 \%} & \textbf{87.08} \% & \textbf{88.98\%} & \textbf{57.23\%} & \textbf{77.80\%} & \textbf{94.45\%} & \textbf{95.83\%} & \textbf{98.26\%} \\
    \bottomrule
    \end{tabular}
    }
    \label{tab:evaluation-lwe-compress}
\end{table*}

\subsubsection{Transciphering/Hybrid HE}
The concept of Hybrid Homomorphic Encryption (HHE) was first introduced by Naehrig et al.~\cite{naehrigCanHomomorphicEncryption2011a}. 
The client encrypts their input using a symmetric encryption scheme, which has an expansion factor of one.
The server, having access to the symmetrically encrypted ciphertext and homomorphically encrypted secret key, can perform a homomorphic decryption of the symmetric ciphertext to get a homomorphic ciphertext of the intended message.
The communication burden of sending a large ciphertext is substituted with a computational effort by the server to perform the conversion.
Recent works have attempted to reduce the computational burden on the server by proposing alternative symmetric encryption schemes that are more \textit{HE-friendly}~\cite{dobraunigPastaCaseHybrid2021, albrechtCiphersMPCFHE2015,dobraunigRastaCipherLow2018,dobraunigPastaCaseHybrid2021, canteautStreamCiphersPractical2018, meauxStreamCiphersEfficient2016}

To summarize, this technique involves conversion from a symmetric encryption scheme to a homomorphic encryption scheme.
While this technique is extremely effective in reducing the upload cost, it can not be used in the opposite direction, from homomorphic ciphertexts to symmetric ciphertexts, to reduce the download cost.
The destination ciphertext must be homomorphic so that the decryption function of the source ciphertext can be computed.

\subsection{Posthoc, Approximate Conversions}
Imprecise conversions are helpful in cases when the result does not need to be operated anymore.
Compression is a prime example.
In posthoc conversion, the destination ciphertext doesn't have to encrypt the same message, as long as the original message can be retrieved, given the new message.
In our compression, we skip the second step of decryption in the source encryption, i.e., the rounding step, which is deferred to the client and only the linear step is performed. 
This way, the compression is done with small computational effort, while also enabling the client to retrieve the correct message.

\subsubsection{Compressing LTV Ciphertexts}
Hu~\cite{huImprovingEfficiencyHomomorphic2013} introduced the concept of \textit{secure converters} for converting between cryptographic schemes.
This is achieved by homomorphically evaluating (part of) the decryption circuit of the source scheme under the destination scheme. 
Within that framework, the author proposed homomorphically converting from LTV ciphertexts to Paillier ciphertexts to reduce bandwidth usage from the server to the client.
The conversion, however, is not precise and the Paillier ciphertexts encrypt a noisy version of the initial plaintexts.
Using this approach, a 256x compression rate is achieved whilst communicating ciphertexts back to the client.
However, the LTV cryptographic scheme is not adopted as a practical homomorphic encryption scheme.

\subsubsection{High-rate Compression}
Brakerski et al.~\cite{brakerskiLeveragingLinearDecryption2019} showed how a high-rate compression, arbitrarily close to one, can be achieved over ciphertexts with the \textit{linear-decrypt-and-multiply} characteristic.
Cryptosystems with linear-decrypt-and-multiply can decrypt to any multiple of the message.
Based on the authors, among prevalent encryption schemes, only GSW falls into that category.
Assuming the goal is to encrypt $\{m_0,m_1,\cdots,m_{\ell-1}\}$, then the compression is done by homomorphically decrypting these messages to $\{m_0+e_0,\Delta m_1+e_1,\cdots,\Delta^{\ell-1}+e_{\ell-1} m_{\ell-1}\}$, where $e_i$'s are noise introduced from the homomorphic cryptosystem, similar to LWE.
By adding these messages together, the server obtains one large plaintext, encrypted under an additive ciphertext which is sends to the client.

\subsubsection{GSW Compression}
Gentry et al.~\cite{gentryCompressibleFHEApplications2019} also proposed a method to compress many GSW ciphertexts into high-rate PVW ciphertexts.
The ratio between the plaintext and ciphertext can be arbitrarily close to zero in their construction.
However, this can only be achieved if the underlying aggregate plaintext is very large.
Specifically, for the ratio to be $1-\epsilon$, the aggregate plaintext must be proportional to $1/\epsilon^3$.
The authors described how to construct a PIR protocol from this technique, but their compression techniques is not applicable to any other type of ciphertext.

\section{Evaluation}
\label{sec:evaluation}

We implemented our compression technique as a library in C++ using GMP.
We use Paillier as the additive encryption scheme, which is extended to Damgard-Jurik when we require a larger plaintext space.
We use a 3072-bit modulus for Pallier, composed of two 1536-bit primes, which is the recommended modulus size for 128-bit security~\cite{barkerRecommendationKeyManagement2020}.
We experiment with LWE and RLWE parameters satisfying 128-bit security but our methods can be applied to other LWE and RLWE parameters without any change.

Our code is open source and available upon acceptance. 
We also integrated it into existing FHE libraries like OpenFHE~\footnote{\url{https://github.com/openfheorg/openfhe-development}} and Concrete~\footnote{\url{https://github.com/zama-ai/concrete}} to show the effectiveness.
We also parallelized our implementation to minimize the latency of the compression.
Specifically, we parallelize over the LWE dimension $n$ or the number of LWE ciphertexts that are compression, depending on whichever is larger.
Using this dynamic approach, we use existing cores on our machines even when compressing few ciphertexts.

\paragraph{Experiment Scenarios.}

We experiment under two scenarios: 1) compressing a single LWE ciphertext or RLWE coefficient 2) compressing many LWE ciphertexts or multiple RLWE coefficients.
The former is useful in applications with small outputs such as private inference, whereas the latter is used for applications with large outputs such as image processing.

\subsection{Single Compression Evaluation}

\Cref{tab:evaluation-lwe-compress} summarizes the results for compressing a single LWE ciphertext. We choose LWE parameters adopted in existing libraries implementing variants of LWE encryption~\cite{zamaConcreteTFHECompiler2022,albadawiOpenFHEOpenSourceFully2022}.
The results show that we consistently provide high compression rates. Notably, for $\log_2 q = 64$, our compression rates are over 84\%.

Similarly, for compression of a single RLWE coefficient, we use common parameters for RLWE-based schemes such as BFV~\cite{brakerskiFullyHomomorphicEncryption2012,fan2012somewhat} and BGV~\cite{brakerskiLeveledFullyHomomorphic2012}, which are used in libraries such as SEAL~\cite{MicrosoftSEALRelease2023}, Lattigo and OpenFHE~\cite{albadawiOpenFHEOpenSourceFully2022}.
Recall that compression is compatible with modulus switching so we choose the parameters corresponding to the lowest level in a BFV/BGV parameter set. We achieve over 85\% compression and up to 98\%.





\subsection{Measuring Compression Key Sizes}

Using the technique described in \Cref{sec:smaller-compression-key}, we can pack the compression key and have the server unpack the key.
The unpacking can be done offline, as soon as the server receives the packed key, to reduce latency during compression.
The compression procedure is identical after the key has been unpacked, so we do not report the runtime of compression again.
Instead, we measure the size of the compression key, with and without packing, and report the time required to unpack the key.
We also distinguish two cases, non-binary and binary keys.
In the case of binary keys, we use $\delta=q+nq$ so more bits of the secret key can fit within the same ciphertext.
\Cref{tab:packed-compression-key} shows the size of the packed compression keys in the two cases. Note that even the size of the unpacked key is much smaller than commonly used cryptographic keys such as relinearization keys, automorphism keys, and bootstrapping keys, which could be as large as 100 MB.

\begin{table}[b]
    \centering
    \caption{
        Size of packed compression keys and unpacking time.
        We distinguish the case of binary and non-binary keys since binary keys can be packed more than non-binary keys.
        The Paillier modulus is 3072 bits in all cases.
    }
    \resizebox{\columnwidth}{!}{
    \begin{tabular}{c|c|c|c|c}
    \toprule
         $(n,\log_2 q)$ & (630,64) & (742,64) & (870,64) & (1305,11) \\
    \midrule
        Unpacked Key & 240 KB & 284 KB & 334 KB & 501 KB \\
    \midrule
        Packed Non-binary Key & 22 KB & 26 KB & 30 KB & 11 KB \\
        Unpacking Time & 14 ms & 25 ms & 74 ms & 15 ms \\
    \midrule
        Packed Binary Key & 12 KB & 14 KB & 16 KB & 7 KB \\
        Unpacking Time & 13 ms & 12 ms & 13 ms & 15 ms \\
    \bottomrule
    \end{tabular}
    }
    \label{tab:packed-compression-key}
\end{table}

\subsection{Batched Compression Evaluation}

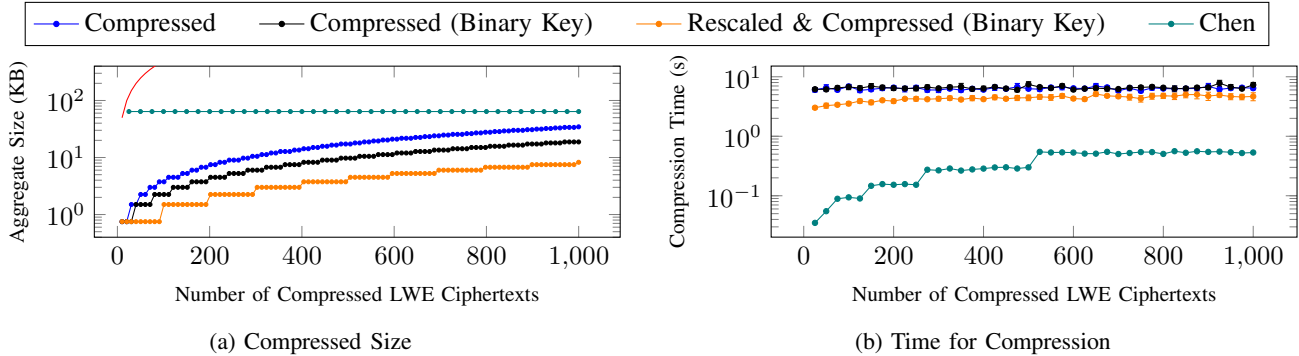
\begin{figure*}[t]
    \centering
    \begin{tikzpicture}
        \begin{axis}[
            height=0.2\columnwidth,
            width=2*\columnwidth, 
            hide axis,
            xmin=0,
            xmax=1,
            ymin=0,
            ymax=1,
            legend columns=-1, 
            legend style={/tikz/every even column/.append style={column sep=0.5cm}},
            legend to name=named, 
        ]
    
            \addlegendimage{color=blue, mark size=1pt, mark=*}
            \addlegendentry{Compressed}
    
            \addlegendimage{color=black, mark size=1pt, mark=*}
            \addlegendentry{Compressed (Binary Key)}

            \addlegendimage{color=orange, mark size=1pt, mark=*}
            \addlegendentry{Rescaled \& Compressed (Binary Key)}
    
            \addlegendimage{color=teal, mark size=1pt, mark=*}
            \addlegendentry{Chen}
        \end{axis}
        \end{tikzpicture}
    
        \ref{named} 
\begin{subfigure}[b]{\columnwidth}
    \centering
    \begin{tikzpicture}[
        declare function={
            logb(\x,\y) = ln(\x)/ln(\y);
            s=1; 
            log2m=3072; 
            n=630; 
            log2q=64; 
            newlog2q=20; 
        },
    ]
    \begin{axis}[
        name=plot1,
        at={(0,0)},
        xlabel = {\footnotesize Number of Compressed LWE Ciphertexts},
        ylabel = {\footnotesize Aggregate Size (KB)},
        domain=0:1000,
        samples=100,
        ymin = 0.4,
        ymax = 400,
        ymode = log,
        xmin = -50,
        xmax = 1090,
        height = 0.45\columnwidth,
        width = \columnwidth,
        legend style={at={(0.2,0.6)},anchor=east} 
    ]

        \addplot [color=red, mark size=0.75pt] {x * (n+1) * log2q / 8192 };
    
        \addplot [color=blue, mark size=0.75pt, mark=*] { (s+1)*log2m * ceil(x / floor(s * log2m / (1 + ceil(logb(n,2)) + 2*log2q))) / 8192 };

        \addplot [color=black, mark size=0.75pt, mark=*] { (s+1)*log2m * ceil(x / floor(s * log2m / (1 + ceil(logb(n,2)) + log2q))) / 8192 };

        \addplot [color=orange, mark size=0.75pt, mark=*] { (s+1)*log2m * ceil(x / floor(s * log2m / (1 + ceil(logb(n,2)) + newlog2q))) / 8192 };

        \addplot [color=teal, mark size=0.75pt, mark=*] table [
            x=slots,
            y expr=\thisrow{compressed_size_B}/1024,
            col sep=comma
        ] {data/chen.csv};
        
    \end{axis}    
    \end{tikzpicture}
    \caption{Compressed Size}
    \label{fig:batched-lwe-compression-n=630-communication}
\end{subfigure}
~
\begin{subfigure}[b]{\columnwidth}
    \begin{tikzpicture}[]
        \begin{axis}[
            name=plot1,
            at={(0,0)},
            xlabel = {\footnotesize Number of Compressed LWE Ciphertexts},
            ylabel = {\footnotesize Compression Time (s)},
            height = 0.45 \columnwidth,
            width = \columnwidth,
            ymode=log,
            legend pos=south east,
            error bars/y dir=both,
            error bars/y explicit
        ]

            \addplot [color=blue, mark size=1pt, mark=*] table [
                x=num_cts,
                y expr=\thisrow{compress_time_us_mean}/1000000,
                y error expr=\thisrow{compress_time_us_std}/1000000,
                col sep=comma
            ] {data/table2_nonbinkey.csv};

            \addplot [color=black, mark size=1pt, mark=*] table [
                x=num_cts,
                y expr=\thisrow{compress_time_us_mean}/1000000,
                y error expr=\thisrow{compress_time_us_std}/1000000,
                col sep=comma
            ] {data/table2_binkey.csv};

            \addplot [color=orange, mark size=1pt, mark=*] table [
                x=num_cts,
                y expr=\thisrow{compress_time_us_mean}/1000000,
                y error expr=\thisrow{compress_time_us_std}/1000000,
                col sep=comma
            ] {data/table2_binkey_switched.csv};

            \addplot [color=teal, mark size=1pt, mark=*] table [
                x=slots,
                y expr=\thisrow{total_time_ms}/1000,
                col sep=comma
            ] {data/chen.csv};
        \end{axis}    
    \end{tikzpicture}
    \caption{Time for Compression}
    \label{fig:batched-lwe-compression-n=630-computation}
\end{subfigure}
    \caption{
        Compressed size and compression time required for compressing LWE ciphertexts with $(n,q)=(630, 2^{64})$ using batched compression. The red line denotes the baseline size of uncompressed LWE ciphertexts.
    }
    \label{fig:compressing-batched-lwe}
\end{figure*}

Next, we evaluate the batched compression of LWE ciphertexts.
In both cases, we use the batched compression algorithm to compress $\ell$ ciphertext into additive Paillier ciphertexts.
Note that for batched compression, we can not use packed compression keys.

We distinguish the case of binary keys from non-binary keys in the experiment, denoted by blue and black lines, respectively.
As mentioned in \Cref{sec:batched-compression}, we set the scale to $\gamma=q + n q^2$ and $\gamma=q + n q$ in the case of non-binary and binary keys, respectively.
We also visualize a case where we rescale to a smaller modulus before compression.

The alternative approach to packing many LWE ciphertexts is RLWE packing~\cite{chenEfficientHomomorphicConversion2021,chillottiTFHEFastFully2020}.
Our reference point for RLWE packing is the work of Chen et al.~\cite{chenEfficientHomomorphicConversion2021} which is state-of-the-art in compression and has better runtime than related work.
This work maps LWE ciphertexts in $\ZZ_q^n\times\ZZ$ to an RLWE ciphertext in $\ZZ_q[X]/(X^n+1)$.

\Cref{fig:batched-lwe-compression-n=630-communication} shows the size after compression as a function of the number of compressed LWE ciphertexts.
\Cref{fig:batched-lwe-compression-n=630-computation} also shows the runtime of batched compression.
The blue and black plots correspond to non-binary and binary keys, respectively.
The orange plot also shows batched compression but over ciphertexts that are rescaled to $r=2^{20}$.
In the case of non-binary keys, we only need one Paillier ciphertext for up to 14 LWE ciphertexts and for binary keys, it is about 26.
Overall, compressing more LWE ciphertexts offers more compression compared to compressing only one LWE ciphertext.
This is because more of the plaintext space of Paillier is being utilized as more ciphertexts are compressed.
We also observe that rescaling improves the runtime and allows more compression, as is expected.
In comparison to the work of Chen et al., our compression protocol offers more than one order of magnitude compression compared to RLWE packing.
However, our compression approach is slower than RLWE packing, particularly due to the use of expensive modular exponentiations.

\section{\protocol{} : PIR using Compression}

In this section, we present \textit{\protocol{}}, our new low-communication PIR protocol.
\protocol{} takes advantage of the compression technique in this work, in combination with techniques in the literature such as layered encryption~\cite{angelPIRCompressedQueries2018a, kiayiasOptimalRatePrivate2015, lipmaaObliviousTransferProtocol2005} and PIR using LWE~\cite{henzingerOneServerPrice2023, davidsonFrodoPIRSimpleScalable2023}.
This results in a PIR protocol with the lowest overall communication costs compared to all related work.
We first describe a simple version of \protocol{} and show how to adapt the construction to work with smaller keys and larger payloads.

\subsection{\protocol{} Description}

We restructure the database as a $(k+1)$-dimensional hypercube and perform PIR using LWE on the first dimension and PIR using Paillier on the subsequent dimensions.
Using LWE in the first dimension improves the runtime significantly, but we add some layers of Paillier to reduce the communication costs.

\Cref{alg:compresspir} shows the detailed description of \protocol{}. We follow the framework of PIR with preprocessing~\cite{liReducingCiphertextExpansion2015} and provide four routines.
The high-level steps for \protocol{} are as follows:
\begin{enumerate}
    \item In the setup phase~(\Cref{compresspir:hint-compute}), the server hint is calculated and stored by the server. This hint can be reused for all queries by any client and updated locally as the database changes.
    \item Upon receiving the query from the client, the server first expands the compression key to use throughout the process~(\Cref{compresspir:expand-key}).
    \item Using $\qu_0$, The server performs PIR using LWE on the first dimension, which is assisted by the hint generated in the setup~(\Cref{compresspir:online-lwe}). This step is similar to that of related work~\cite{henzingerOneServerPrice2023,davidsonFrodoPIRSimpleScalable2023}.
    \item The output of this step is rescaled and compressed using the compression functions and the expanded key to get Paillier ciphertexts.
    \item The Paillier ciphertexts, $\qu_1,\cdots, \qu_{k}$, are then used to do PIR using layered encryption, expanding the size by a factor of two in each layer. The result is sent to the client~(\Cref{compresspir:start-paillier-pir}-\ref{compresspir:paillier-pir-1}).
    \item Upon receiving the response, the client decrypts the Paillier ciphertexts and finally performs a modified LWE decryption to retrieve the response~(\Cref{compresspir:start-paillier-decrypt}-\ref{compresspir:paillier-merge}).
\end{enumerate}

\begin{algorithm}[]
\caption{
  Complete description of \protocol{}. $q$ is the LWE ciphertext modulus, and $r$ is the smallest divisor of $q$ such that $r\geq 2(n+1)p$, $p$ is the plaintext modulus and $\Delta = q/p$.
  Database $\db\in\ZZ_p^{N_0\times d_0}$ where $N=d_0 d_1 \cdots d_k$ and $N_{\ell} = N / (d_0 d_1 \cdots d_{\ell})$ for all $\ell$. Also, as setup $\textbf{A} \sample \ZZ_q^{d_0 \times n}$. Paillier ciphertexts are in $\ZZ_{m^2}$. We denote Paillier homomorphic addition and scalar multiplication by $\oplus$ and $\otimes$, respectively.
}
	 \label{alg:compresspir}
	 \begin{algorithmic}[1]
    \Procedure{Setup}{$\db \in \ZZ_q^{N_0 \times d_0}$}
        \State $\hint = \db \times \textbf{A} \in \ZZ_q^{N_0 \times n}$ \label{compresspir:hint-compute}
        \State \Return $\hint$
    \EndProcedure
    \vspace{3mm}
    \Procedure{Query}{$(i,i_0)\in [N_0]\times[d_0]$}
        \State Generate Paillier keys $(\paillierpk, \paillierkey)$
        \State Sample LWE key $\lwekey\sample\{0,1\}^{n}$
        \State $\ck \leftarrow \textsc{PaillierEncrypt}(\paillierkey, \sk)$ \label{compresspir:compression-key}
        \For {$\ell \in \{1,2,\cdots,k\}$}
            \State $i_{\ell} = \floor{i/N_{\ell}} \mod d_{\ell}$
        \EndFor
        \State $u_j =$ selection vector for index $i_j$, $j\in[k+1]$
        \State Sample $e\leftarrow \chi^n$
        \State $\qu_0 = \textbf{A}\cdot \lwekey + e + \Delta\cdot u_0$ \Comment{$\qu_0\in\ZZ_q^{d_0}$}\label{compresspir:lwe-encrypt}
        \For {$\ell \in \{1,2,\cdots,k\}$}
                \State $\qu_{\ell} = \textsc{PaillierEnc}(\paillierkey, u_{\ell})$ 
            \Comment{$\qu_{\ell}\in\ZZ_{m^2}^{d_{\ell}}$}
                \label{compresspir:paillier-encrypt}
        \EndFor
        \State \Return $(\paillierkey, (\paillierpk, \ck, {\qu}_0, {\qu}_1, \cdots, {\qu}_k))$
    \EndProcedure
    \vspace{3mm}
    \Procedure{Response}{$\db, \hint,\qu=(\paillierpk, \ck,{\qu}_0, {\qu}_1, \cdots, {\qu}_k)$}
        \State $\eck \leftarrow \textsc{ExpandCompressionKey}_{r}(\ck)$ \label{compresspir:expand-key}
        \State $b = \db \cdot {\qu}_0$ \label{compresspir:online-lwe}
        \State $\textbf{D} = [\hint | b] \in \ZZ_r^{N_0 \times (n+1)}$
        \State $\textbf{D}_0\leftarrow$ Rescale elements in $\textbf{D}$ to modulus $r$ \label{compresspir:rescale-online}
        \For {$i\in [N_0]$}
            \Comment{$c_0\in\ZZ_{m^2}^{N_0\times 1}$}
            \State $c_0[i] = \textsc{FastLWECompress}_{r}(\eck, \textbf{D}_0[i])$
            \label{compresspir:fastcompress}
        \EndFor
        \For {$\ell = \{1, \cdots, k\}$} \label{compresspir:start-paillier-pir}
            \State Initialize $c_{\ell}$ with zeros
            \Comment{$c_{\ell} \in \ZZ_{m^2}^{N_{\ell}\times 2^{\ell}}$}
            \For {$t\in [N_{\ell}]$}
                \For {$h \in [2^{\ell-1}]$}
                    \For {$j \in [d_{\ell}]$}
                        \State $u_{0j} \leftarrow c_{\ell-1}[jN_{\ell}+t][h] \mod m$
                        \State $u_{1j} \leftarrow \floor{c_{\ell-1}[jN_{\ell}+t][h] / m}$
                        \State $c_{\ell}[t][2h] \leftarrow c_{\ell}[t][2h] \oplus (\qu_{\ell}[j] \otimes u_{0j}) $ \label{compresspir:paillier-pir-0}
                        \State $c_{\ell}[t][2h+1] \leftarrow c_{\ell}[t][2h+1] \oplus (\qu_{\ell}[j] \otimes u_{1j})$ \label{compresspir:paillier-pir-1}
                    \EndFor
                \EndFor
            \EndFor
        \EndFor
        \State \Return $c_k[0] \in \ZZ_{m^2}^{2^k}$
    \EndProcedure
    \vspace{3mm}
    \Procedure{Extract}{$\st = \paillierkey, f_k \in \ZZ_{m^2}^{2^k}$}
        \For {$\ell \in \{k, k-1, \cdots, 1\} $}
        \label{compresspir:start-paillier-decrypt}
            \State $p_{\ell} = \textsc{PaillierDec}(\paillierkey, f_{\ell})$
            \Comment{$p_{\ell}\in\ZZ_{m}^{2^{\ell}}$}
            \For {$h\in[2^{\ell-1}]$}
                \State $f_{\ell-1}[h] = m\cdot p_{\ell}[2h+1] + p_{\ell}[2h]$
                \label{compresspir:paillier-merge}
            \EndFor
        \EndFor
        \State $f=\textsc{ModifiedLWEDecrypt}_{r, p}(\paillierkey, f_0)$
        \label{compresspir:lwe-dec}
        \State\Return $f$
    \EndProcedure    
  \end{algorithmic}
\end{algorithm}

\begin{theorem}[\protocol{}]
    For LWE parameters $(n, q, \chi_e, \chi_s)$, assume $\chi_e$ is a discrete Gaussian with standard deviation $\sigma$ and $\chi_s$ is a binary distribution. Assume these parameters are $\epsilon_{L}$-secure for LWE with $d_0$ samples and take plaintext modulus $p$ such that $p|q$ and 
    \begin{align}
        q/p > 2p\sigma \sqrt{2d_0\ln (2/\delta)} 
        \label{eq:correctness}
    \end{align}
    Also, assume we instantiate the Paillier cryptosystem with modulus $m$ such that it is $\epsilon_{P}$-secure and $m > q+nq$. Then for a random matrix $\textbf{A}\in\ZZ_{q}^{d_0\times n}$, \protocol{} is a $2(\epsilon_{L}+\epsilon_{P})$-secure PIR scheme on database of size $N$ with items in $\ZZ_p$, with $1-\delta$ success rate.
\end{theorem}

The complete proof of correctness and security are provided in \appsection{appendix:proof}.

\subsubsection{Concrete Costs of \protocol{}}
\label{sec:concrete-costs}

The concrete number of operations in each routine in \protocol{}, along with the party that must perform those operations is listed below.
\begin{itemize}
    \item $\textsc{Setup}$ (Server): $N_0 d_0 n$ multiplications and additions in $\ZZ_q$ (\Cref{compresspir:hint-compute})
    \item $\textsc{Query}$ (Client): $d_0$ LWE encryptions (\Cref{compresspir:lwe-encrypt}) and $n+\sum_{\ell=1}^{k} d_{\ell}$ Paillier encryptions (\Cref{compresspir:compression-key} and \Cref{compresspir:paillier-encrypt})
    \item $\textsc{Response}$ (Server): $N=N_0d_0$ multiplications and additions in $\ZZ_q$ (\Cref{compresspir:online-lwe}), and $n\log r+\frac{1}{2} n\log_2 r$ multiplications in $\ZZ_{m^2}$, i.e., Paillier additions (\Cref{compresspir:expand-key} and \Cref{compresspir:fastcompress}) and $\sum_{\ell\in[k]} 2^{\ell} N_{\ell}$ exponentiations in $\ZZ_{m^2}$, i.e., Paillier scalar multiplications (\Cref{compresspir:paillier-pir-0} and \Cref{compresspir:paillier-pir-1}).
    \item $\textsc{Extract}$ (Client): $2^{k+1}$ Paillier decryptions
\end{itemize}

Similarly, the concrete communication costs of the protocol are listed below. 
\begin{itemize}
    \item Client to Server: $d_0 \log_2 q + 2 \log_2 m (n + \sum_{\ell=1}^{k} d_i)$
    \item Server to Client: $2^{k+1}\log_2 m$
\end{itemize}

\subsection{Updating the Hint} 
When the database changes, the hint must also be updated, but the hint can be updated locally by the server and does not require any communication with the clients.
This is in contrast to other works which rely on hints that require sending updates to the clients~\cite{henzingerOneServerPrice2023, davidsonFrodoPIRSimpleScalable2023}.
Moreover, small changes to the database can be handled with small changes to the hint to reduce the computation cost.
For example, assume $\db'$ is an updated database compared to $\db$. If we denote the hint for the new database by $\hint'$, then the relationship between the previous hint and the new hint would be as follows
\begin{align}
    \hint' = \mask\cdot\db' = \mask\cdot(\db +\db_{\Delta}) = \hint + \mask\cdot\db_{\Delta}
\end{align}
where $\db_{\Delta}$ is the difference between the two databases.
Assuming the change is small, many columns of $\db_{\Delta}$ will be zero. For example, assume that only column $k$ of $\db_{\Delta}$ has non-zero elements, then we only need to calculate $\mask\cdot\db_{\Delta}[:,k]$, which is a matrix-vector multiplication with only $d_0n$ multiplications and additions in $\ZZ_q$.

\subsection{Smaller Keys or Compressed Large Payloads}
\protocol{} can be modified in one of two ways to either reduce the size of the compression keys or produce a more compressed response.
The changes required for these two modifications can not combined, so we propose two variants of \protocol{} which we denote \protocolsingle{} (with smaller compression keys) and \protocolbatched{} (with more compressed responses).
We provide a high-level description of these modifications and leave the full detailed description for the full version.

In \protocolsingle{}, we use the technique from \Cref{sec:smaller-compression-key} to produce packed compression keys.
The server must then unpack the compressed keys before using them.
More precisely, in \Cref{alg:compresspir}, we change \Cref{compresspir:compression-key} to encrypt the key in a compressed way using \textsc{GeneratePackedKey} from \Cref{alg:packed-key-compress}.
We also add a step before \Cref{compresspir:expand-key} to unpack the key using \textsc{UnpackCompressionKey} and change the \Cref{compresspir:lwe-dec} to the corresponding function with packed keys, \textsc{ModifedLWEDecryptPackedKey}.

In \protocolbatched{}, we adapt the protocol to produce more compressed responses when the payload is large, which is done using the batched compression technique.
The high-level description of \protocolbatched{} is as follows:
Assume each database element is an element in $\ZZ_p^{\ell}$ for some $\ell\in\NN$.
Let $\db^j$ denote a database consisting of the $j^{th}$ component of all elements in this database.
Corresponding to this, we generate $\hint^j$, $\textbf{D}^j$, and $\textbf{D}_{0}^{j}$, as is done in \Cref{alg:compresspir}.
We replace \Cref{compresspir:fastcompress} to perform a batched compression in the following manner
\begin{align}
    \textsc{FastBatchedLWECompress}_{r, \gamma}(\eck, \{\mathbf{D}_0^{j}\}_{j\in[\ell]})
\end{align}

and proceed with the rest of the protocol as before.
We also change the decryption function in \Cref{compresspir:lwe-dec} to the corresponding decryption for modified batched decryption.

\subsection{Additional techniques for \protocol{}}

In addition to the techniques mentioned in the previous two section, we use two more techniques to further reduce communication costs.
Firstly, we can use the technique from \Cref{sec:overlapping-noise} in one of two ways 1) Pack more LWE ciphertexts within each Paillier ciphertext to produce a smaller response or 2) Pack more bits of the secret key with each Paillier ciphertext to have a smaller compression key.
As mentioned, to produce correct results using this technique, the error in the LWE ciphertext must be small enough, so we must consider this constraint for correctness.
This changes the correctness condition in \Cref{eq:correctness} to 
\begin{align}
    \label{eq:correctness-modified}
    q/p > 4p\sigma \sqrt{2d_0\ln (2/\delta)} 
\end{align}
which is what we use in our experimental evaluation.
We also use a technique proposed by Beck~\cite{beckRandomizedDecryptionRD2015} to reduce the size of uploaded Paillier ciphertexts, at the cost of small computational overhead for the server.
Due to space restrictions, we provide the proof of correctness and security of \protocol{} using these techniques in the full version of the paper.

\section{PIR Evaluation}

For our evaluation, we first detail the process for selecting the parameters of \protocol{}.
Given the many parameters that must be chosen, this amounts to a non-trivial optimization problem.
After that, we provide runtimes for \protocol{} to demonstrate the scalability. Finally, we compare with related work on PIR with no setup such as WhisPIR, HintlessPIR, and YPIR.
Our results demonstrate that \protocol{} introduces a new category of PIR protocols with low communication that 

\subsection{Parameter Selection for \protocol{}}
The LWE parameters directly effect the performance of \protocol{}.
However, the set of secure LWE parameters is large and experimenting with all parameters set is infeasible.
Hence, we instantiate \protocol{} with three LWE parameters $(n,q,\chi_e, \chi_s)$ which are representative of different tradeoffs.
A smaller $n$ and $q$ results in fewer operations in based on the analysis of \Cref{sec:concrete-costs}, but limit the choice of $p$.
In contrast, higher $n$ and $q$ allow for a larger $p$ and $d_0$.
We choose three $(n,q)$ pairs, and let $\chi_e$ and $\chi_s$ be a discrete Gaussian with standard deviation $\sigma = 6.4$ and uniform binary distribution, respectively.
While in other works based on LWE~\cite{henzingerOneServerPrice2023,davidsonFrodoPIRSimpleScalable2023}, $q$ is chosen as a power of two, e.g., $2^{32}$ or $2^{64}$, to leverage the native CPU word size, we opt for a small $q$ since the bitlength of $q$ determines the total communication cost and number of operations.
Our chosen parameters for LWE provide 128-bit security based on the lattice estimator~\cite{albrechtConcreteHardnessLearning2015}.
We set the failure rate to $\delta=2^{-40}$ and for every $p$ find the upper bound for $d_0$ based on the correctness constraint, \Cref{eq:correctness}.
The upper bound on $d_0$ for each parameter set and each value of $p$ is shown in \Cref{tab:d0-bound}.
For the Paillier cryptosystem, we use a 3072-bit modulus which provides 128-bit security~\cite{barkerRecommendationKeyManagement2020}.


\begin{table}[H]
    \centering
    \caption{Upper bound on $d_0$ for every value of $p$, based on \Cref{eq:correctness-modified}. In all cases, $\sigma=6.4$, and $\delta=2^{-40}$. Dashes indicate cases where no $d_0$ satisfies the equation.}
    \label{tab:d0-bound}
    \begin{tabular}{c|c|c|c}
    \toprule
         \diagbox{$p$}{$(n,q)$} & $(630, 17)$ & $(840, 22)$ & $(1023, 27)$ \\
    \midrule
        $2^1$ & 28825 & 29517568 & 30225990335 \\
        $2^2$ & 1801 & 1844848 & 1889124395 \\
        $2^3$ & 112 & 115303 & 118070274 \\
        $2^4$ & 7 & 7206 & 7379392 \\
        $2^5$ & - & 450 & 461212 \\
        $2^6$ & - & 28 & 28825 \\
    \bottomrule
    \end{tabular}
\end{table}

The two important metrics to evaluate performance are total communication cost and server online runtime.
The remainder of the parameters, such as the dimensions of the database $\{d_i\}$, $p$, and the choice of the protocol (\protocolsingle{} or \protocolbatched{}) are chosen to minimize these costs.
However, the selection of these parameters is a non-trivial optimization problem.
While the communication cost can be derived with a closed-form formula, this is not the case for the server online runtime.
Moreover, many parameter sets fall on the Pareto frontier, i.e., are dominant in either communication or computation.
Hence, for a fixed number of rows and payload size, we aim to find as many parameter sets that fall on the Pareto frontier.
For this, we iterate over the list of all parameter sets and maintain a list which is Pareto optimal.
To narrow down the search space, represent each parameter set by the number of operations in the different steps, as calculated in \Cref{sec:concrete-costs}.
Parameter sets which are worse than another parameter set in all steps are trivially excluded.
Moreover, we use logistic regression to predict if a parameter set A dominates another parameter set B, given the concrete number of operations in the steps of the protocol.
This further reduces the space of parameter sets to a manageable size, which we can run experimentally.

\subsection{Performance of \protocol{}}
\Cref{fig:zippir-costs-per-num-rows} visualizes the communication and computation cost of \protocol{} as a function of the database size, with the goal of retrieving at least one bit from the database.
Using the procedure described in the previous section, we maintain the list of Pareto optimal parameters for each database size and plot them.
For each database size, each point on the communication graph corresponds to a point on the runtime graph, i.e., the point with higher communication has lower computation and vice-versa.
Within the communication graph, we also include the minimum required communication for related work on low-communication PIR.

There are several important observations from this graph.
Firstly, we observe that \protocolsingle{} is the best option, given that the requested payload is small.
Cases which \protocolbatched{} only appear when the payload size is large.
Second, for small databases, the communication cost of \protocol{} is small, as opposed to all other protocols in the literature, which have a lower bound on communication due to the use of large cryptographic keys.
Lastly, the minimum communication cost of \protocol{} grows sublinearly, roughly proportional to $|\db|^{0.27}$, which demonstrates the scalability of the protocol.

\newcommand{\myScatterClasses}{
    scatter/classes={
        small single={mark=*,mark size=0.6pt,yellow},
        medium single={mark=*,mark size=0.8pt,orange},
        large single={mark=*,mark size=1pt,red},
        small batched={mark=*,mark size=0.6pt,white},
        medium batched={mark=*,mark size=0.8pt,white},
        large batched={mark=*,mark size=1pt,white}
    }
}

\newcommand{\protocolsinglesize}{1.5pt}
\newcommand{\protocolsinglecolor}{red}

\newcommand{\protocolbatchedsize}{1.5pt}
\newcommand{\protocolbatchedcolor}{blue}

\newcommand{\hintlesssize}{1.5pt}
\newcommand{\hintlesscolor}{cyan}

\newcommand{\whispirsize}{1.5pt}
\newcommand{\whispircolor}{teal}

\newcommand{\ypirsize}{1.5pt}
\newcommand{\ypircolor}{green}

\begin{figure}[t]
    \centering
    \begin{subfigure}{\columnwidth}
        \centering
        \begin{tikzpicture}
            \begin{axis}[
                ylabel={\footnotesize Online Time (s)},
                \myScatterClasses,
                width=\textwidth,
                height=0.4\textwidth,
                xmode=log,
                log basis x=2,                   
                ymode=log,
                log basis y=10,                   
                xtick={
                        8388608, 33554432, 134217728, 536870912, 2147483648, 8589934592, 34359738368, 137438953472, 549755813888, 2199023255552    
                },
                xticklabels={}
            ]
            \pgfplotstableread[col sep=comma]{data/Total_Size_vs_N_s1.csv}\datatable
            \addplot[
                scatter, 
                only marks,
                scatter src=explicit symbolic
            ] table [meta=mode, x=N input, y=Server Online Time] {\datatable};
            \end{axis}
        \end{tikzpicture}
    \end{subfigure}
    ~
    \begin{subfigure}{\columnwidth}
        \centering
        \begin{tikzpicture}
            \begin{axis}[
                xlabel={Database Size},
                ylabel={\footnotesize Total Size (KB)},
                \myScatterClasses,
                width=\textwidth,
                height=0.6\textwidth,
                xmode=log,
                log basis x=2,                   
                ymode=log,
                log basis y=10,       
                xtick={
                    16777216,
                    134217728,
                    1073741824,
                    8589934592,
                    34359738368,
                    137438953472,
                    4398046511104
                    },
                xticklabels={
                    2MB,
                    16MB,
                    128MB,
                    1GB,
                    4GB,
                    16GB,
                    524288
                },
                tick label style={font=\footnotesize}, 
                legend pos = north west
            ]
            
            \addplot[mark=none, \hintlesscolor, samples=2, domain=8000000:160000000000] {387};
            \addlegendentry{\scriptsize HintlessPIR (Lower Bound)}

            \addplot[mark=none, \whispircolor, samples=2, domain=8000000:160000000000] {441};
            \addlegendentry{\scriptsize WhisPIR (Lower Bound)}

            \addplot[mark=none, \ypircolor, samples=2, domain=8000000:160000000000] {846};
            \addlegendentry{\scriptsize YPIR (Lower Bound)}
            
            \pgfplotstableread[col sep=comma]{data/Total_Size_vs_N_s1.csv}\datatable
            \addplot[
                scatter, 
                only marks,
                scatter src=explicit symbolic
            ] table [meta=mode, x=N input, y=Total Size (KB)] {\datatable};

            
            \end{axis}
        \end{tikzpicture}
    \end{subfigure}
    \caption{
        Communication cost and Server Online Runtime as a function of the database size. Each point in the upper graph as a corresponding point in the lower graph.
        We also plot the minimum communication required for other protocols in the literature.
        The yellow, orange, and red points correspond to $n=(630,17),(840,22),(1023, 27)$, respectively.
    }
    \label{fig:zippir-costs-per-num-rows}
\end{figure}
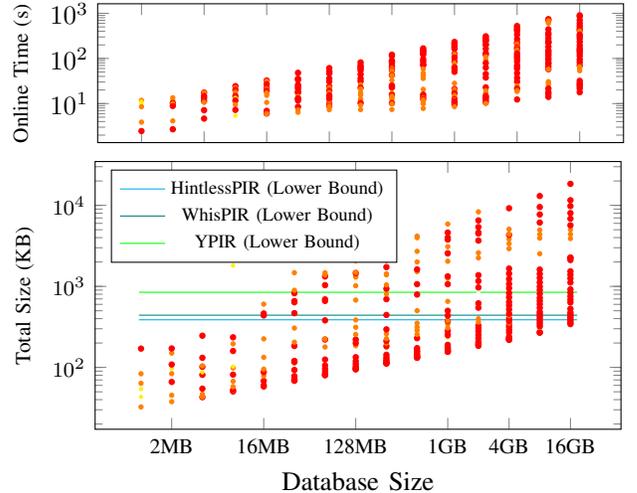

\subsection{Evaluating PIR with No Setup}
We compare \protocol{} with other PIR protocols without setup by measuring communication and computation costs for different database sizes, with the goal of retrieving at least one bit.
We report our measurements in \Cref{fig:eval-pir-comm-comp} in four graphs, corresponding to four different database sizes.
For each database size, we include points corresponding to related work such as HintlessPIR, WhisPIR, and YPIR. Other PIR protocols have communication costs that are much higher than these works.

From these graphs, we can make the following observations.
\protocol{} offers a low communication alternative to existing work, such that in some configurations, we require less than 25\% of the total communication cost of related work.
However, this low communication comes with higher computation costs, which can be addressed in future work.


\begin{figure*}[t]
    \centering
    \begin{tikzpicture}
        \begin{axis}[
            height=0.2\columnwidth,
            width=2*\columnwidth, 
            hide axis,
            xmin=0,
            xmax=1,
            ymin=0,
            ymax=1,
            legend columns=-1, 
            legend style={/tikz/every even column/.append style={column sep=0.5cm}},
            legend to name=namedtradeoff, 
        ]
            \addlegendimage{color=\protocolsinglecolor, mark size=2pt, mark=*}
            \addlegendentry{\protocolsingle{}}
    
    
            \addlegendimage{color=\hintlesscolor, mark size=2pt, mark=*}
            \addlegendentry{HintlessPIR}
    
            \addlegendimage{color=\whispircolor, mark size=2pt, mark=*}
            \addlegendentry{WhisPIR}

            \addlegendimage{color=\ypircolor, mark size=2pt, mark=*}
            \addlegendentry{YPIR}
            
        \end{axis}
        \end{tikzpicture}
    
        \ref{namedtradeoff} 

    \begin{subfigure}{0.24\textwidth}
        \centering
        \begin{tikzpicture}
            \begin{axis}[
                xlabel={\footnotesize Server Online Time (s)},
                ylabel={\footnotesize Total Size (KB)},
                \myScatterClasses,
                width=1.1\textwidth,
                height=0.75\textwidth,
                ymode=log,
                log basis y=2,   
                ymax=2000,
                tick label style={font=\footnotesize}, 
            ]
            \pgfplotstableread[col sep=comma]{data/Total_Size_vs_Server_Online_Time_2147483648x1=256.00MB_all.csv}\datatable
            \addplot[
                scatter, 
                only marks,
                scatter src=explicit symbolic
            ] table [meta=mode, x=Server Online Time, y=Total Size (KB)] {\datatable};
            
            \addplot[
                scatter,
                only marks,
                mark=*, mark size=\hintlesssize, \hintlesscolor
            ] coordinates {
                (0.575,1260)
            };

            \end{axis}
        \end{tikzpicture}
        \caption{$|\db|=$ 0.25 GB}
    \end{subfigure}
    \hfill
    \begin{subfigure}{0.24\textwidth}
        \centering
        \begin{tikzpicture}
            \begin{axis}[
                xlabel={\footnotesize Server Online Time (s)},
                \myScatterClasses,
                width=1.2\textwidth,
                height=0.75\textwidth,
                ymode=log,
                log basis y=2,   
                ymax=2000,
                tick label style={font=\footnotesize}, 
            ]
            \pgfplotstableread[col sep=comma]{data/Total_Size_vs_Server_Online_Time_4294967296x1=512.00MB_all.csv}\datatable
            \addplot[
                scatter, 
                only marks,
                scatter src=explicit symbolic
            ] table [meta=mode, x=Server Online Time, y=Total Size (KB)] {\datatable};
            
            \addplot[
                scatter,
                only marks,
                mark=*, mark size=\hintlesssize, \hintlesscolor
            ] coordinates {
                (0.768,1639)
            };

            \end{axis}
        \end{tikzpicture}
        \caption{$|\db|=$ 0.5 GB}
    \end{subfigure}
    \hfill  
    \begin{subfigure}{0.24\textwidth}
        \centering
        \begin{tikzpicture}
            \begin{axis}[
                xlabel={\footnotesize Server Online Time (s)},
                \myScatterClasses,
                width=1.2\textwidth,
                height=0.75\textwidth,
                ymode=log,
                log basis y=2,   
                xmax=100,
                ymax=2450,
                tick label style={font=\footnotesize}, 
            ]
            \pgfplotstableread[col sep=comma]{data/Total_Size_vs_Server_Online_Time_8589934592x1=1024.00MB_all.csv}\datatable
            \addplot[
                scatter, 
                only marks,
                scatter src=explicit symbolic
            ] table [meta=mode, x=Server Online Time, y=Total Size (KB)] {\datatable};
            \addplot[
                scatter,
                only marks,
                mark=*, mark size=\whispirsize, \whispircolor
            ] coordinates {
                (1.691,564)
                (1.077,620)
                (1.037,878)
                (0.814,1044)
            };
            
            \addplot[
                scatter,
                only marks,
                mark=*, mark size=\hintlesssize, \hintlesscolor
            ] coordinates {
                (1.03,2131)
            };

            \addplot[
                scatter,
                only marks,
                mark=*, mark size=\ypirsize, \ypircolor
            ] coordinates {
                (0.428,858)
            };
            \end{axis}
        \end{tikzpicture}
        \caption{$|\db|=$ 1 GB}
    \end{subfigure}
    \hfill
    \begin{subfigure}{0.24\textwidth}
        \centering
        \begin{tikzpicture}
            \begin{axis}[
                xlabel={\footnotesize Server Online Time (s)},
                \myScatterClasses,
                width=1.2\textwidth,
                height=0.75\textwidth,
                ymode=log,
                log basis y=2,   
                ymax=2500,
                tick label style={font=\footnotesize}, 
            ]
            \pgfplotstableread[col sep=comma]{data/Total_Size_vs_Server_Online_Time_68719476736x1=8192.00MB_all.csv}\datatable
            \addplot[
                scatter, 
                only marks,
                scatter src=explicit symbolic
            ] table [meta=mode, x=Server Online Time, y=Total Size (KB)] {\datatable};
            \addplot[
                scatter,
                only marks,
                mark=*, mark size=\whispirsize, \whispircolor
            ] coordinates {
                (8.412,852)
                (7.432,963)
                (5.984,1183)
                (6.254,2047)
            };            

            \addplot[
                scatter,
                only marks,
                mark=*, mark size=\hintlesssize, \hintlesscolor
            ] coordinates {(2.3, 2128)};

            \addplot[
                scatter,
                only marks,
                mark=*, mark size=\ypirsize, \ypircolor
            ] coordinates {(0.992, 1512)};
            
            \end{axis}
        \end{tikzpicture}
        \caption{$|\db| = $ 8 GB}
    \end{subfigure}
    \caption{Communication cost vs. server online runtime for several database sizes. The red, orange, and yellow points correspond to ZipPIR with different parameters, as described \Cref{fig:eval-pir-comm-comp}.}
    \label{fig:eval-pir-comm-comp}
\end{figure*}
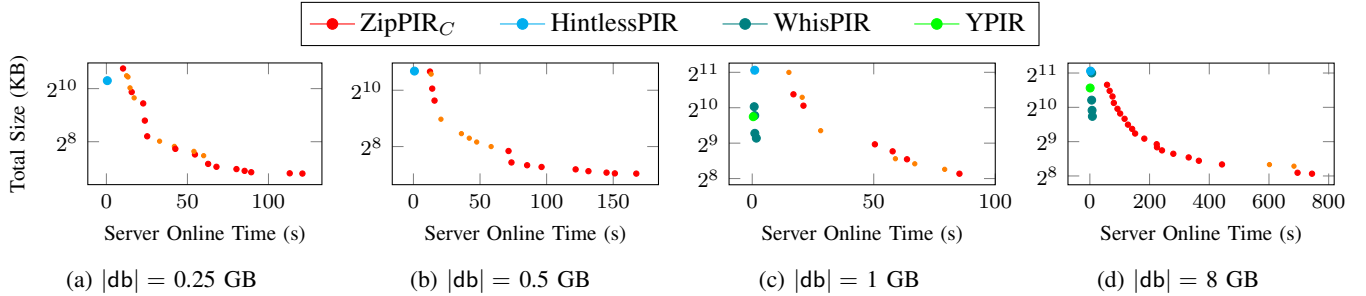

\section{Related Work on PIR}
Computational PIR (CPIR) protocols follow one of three approaches:
1) the server gives a \textit{hint} to the client 2) the client sends cryptographic keys to the server 3) there is no setup, hint, or apriori key exchange.
We describe each approach briefly, the advantages and disadvantages of each approach and list related work.

\subsection{Hint-based PIR}
One approach is for the server to generate a database-dependant hint which is transmitted to the client before the query is issued.
The objective of the hint is to speed up subsequent queries.
SimplePIR~\cite{henzingerOneServerPrice2023} and FrodoPIR~\cite{davidsonFrodoPIRSimpleScalable2023} are two recent works that propose a PIR protocol based on LWE with a client-independent hint.
The hint size is $O(\sqrt{N}n)$ for $N$ database rows and LWE dimension of $n$.
All clients use the same hint which helps respond quickly to PIR queries and achieve very high throughput (up to 10 GB/s).
However, the hint is a high upfront cost (100 MB for a 1 GB database) and must be recalculated and redistributed to the clients every time the database is updated.
The authors show how to update the client hint with a small amount of communication.
DoublePIR extends SimplePIR so that the hint that must be sent to the client is smaller but the overall throughput is less.
In recent work, Henzinger et al. used an improved version of SimplePIR for a private web search application to avoid sending a large hint to the client in a method similar to our work, but with the use of an RLWE-based cryptosystem~\cite{henzingerPrivateWebSearch2023}.

\subsection{PIR with Setup}
Another category of works assumes auxiliary information is sent before the start of the protocol, usually in the form of cryptographic keys.
The cost of sending these keys is amortized over many queries but requires per-client storage on the server.
While this approach is good if there is an established connection between the client and server, it is a high upfront cost.
Moreover, the public keys allow the server to correlate different queries that the client makes so it is not suitable to combine with anonymity networks.
Henzinger et al. also showed that such long-term persistent keys expose the client to state-recovery attacks that could compromise past queries.
Despite these disadvantages, the online time in such protocols is very small and if sufficient queries are made, the runtime and communication cost of setup is amortized.

Works that follow this model include SealPIR~\cite{angelPIRCompressedQueries2018a}, MulPIR~\cite{Ali2019CommunicationComputationTI}, OnionPIR~\cite{mugheesOnionPIRResponseEfficient2021}, Constant-weight PIR~\cite{mahdaviConstantweightPIRSingleround2022}, Pantheon~\cite{ahmadPantheonPrivateRetrieval2022}, FastPIR~\cite{ahmadAddraMetadataprivateVoice2021}, Spiral (and its variants)~\cite{menonSPIRALFastHighRate2022}, and SparsePIR~\cite{patelDonBeDense2023}.

\subsection{PIR without Hints or Setup}
The previous approaches require an established connection between a client and server to amortize the cost of the hint or cryptographic keys across many queries.
For applications where the client only performs a few queries, previous solutions are impractical
Naively applying previous solutions would require the cryptographic keys to be sent as part of the query, resulting in large queries.
Hence, the third approach is to design a PIR protocol that does not require precomputed hints or large cryptographic keys.
Our work also falls in this category and achieves the lowest total communication cost of all PIR protocols in the literature.

HintlessPIR~\cite{liHintlessSingleServerPrivate2023} is a protocol which expands on SimplePIR to remove the need to send the hint.
In short, HintlessPIR retrieves the necessary row of the hint from the server, essentially delegating the step which requires the hint to the server.
YPIR~\cite{menonYPIRHighThroughputSingleServer2024} also takes a similar approach and retrieves the necessary row of the hint using high-rate RLWE ciphertexts.

WhisPIR~\cite{castroWhisPIRStatelessPrivate2024}, on the other hand, expands on the protocols with setup and aims to reduce the number of required cryptographic keys.
The authors propose a PIR protocol focused on being stateless, i.e., working well for ephemeral clients and having low communication.
Two main contributions of WhisPIR are
1) modifications to reduce the number of cryptographic keys that are required 2) not performing relinearization after homomorphic multiplications.
Using these techniques along with a careful choice of parameters, WhisPIR achieves a communication cost that is smaller than related work.

\section{Conclusion}
In this work, we proposed a method for reducing server response sizes in client-server protocols using homomorphic encryption.
Specifically, we showed how to compress LWE ciphertexts, sent from the server to the client, up to 90\% for single ciphertexts and 99\% for many ciphertexts.
Using our compression technique, we proposed \protocol{}, a low-communication PIR protocol, suitable for ephemeral clients and low-latency networks. 
We evaluated both our compression technique and \protocol{} and showed that \protocol{} can query large databases with only 200-500KB of communication.



\bibliographystyle{plain}
\bibliography{references}

\begin{appendices}

\section{LWECompress Using Packed Compression Keys}
\label{sec:lwe-compress-packed-keys}

The necessary procedures for using a packed compression key is given in \Cref{alg:packed-key-compress}. The $\textsc{GeneratePackedKey}$ procedure generates the packed key from the LWE secret key.
The unpacking procedure computes the compression key from the packed compression key by scaling the packed key with different values.
By packing in this particular manner, the compression procedure can be done similar to before, without any changes.
The final change is made in the decryption.
The response is not necessarily in the lower order bits of the additive ciphertext ciphertext anymore, so a division is required before continuing with the rest of the LWE decryption procedure.

\begin{algorithm}[t]
	 \caption{Procedures for using a packed compression key,       including generating the packed key, unpacking it, and the corresponding modified decryption function.
        }
	 \label{alg:packed-key-compress}
	 \begin{algorithmic}[1]
    \Procedure{GeneratePackedKey}{$\addkey,\sk$}
        \State $t = \floor{\frac{0.5\log_2 m}{\log_2\delta}}$
        \For {$i\in[\ceil{n/t}]$}
            \State $r \leftarrow \delta^{-(t-1)} (\sum_{j\in[t]}\sk[it+j] \cdot \delta^{j}) \mod m$
            \State $ \pck_{i} \leftarrow \texttt{AEnc}(\addkey, r)$  
        \EndFor
        \State \Return $\pck$
    \EndProcedure
    \vspace{3mm}
    \Procedure{UnpackCompressionKey$_q$}{$\pck$}
        \For {$i\in[\ceil{n/t}]$}
            \For {$j\in[t]$}
                \State $\ck[it+j] \leftarrow \delta^{t-1-j} \otimes \pck[i]$
            \EndFor
        \EndFor
    \State \Return $\ck$
    \EndProcedure
    \vspace{3mm}
    \Procedure{ModifiedLWEDecryptPackedKey$_{q, p}$}{$\addkey,x$}
        \State $y \leftarrow \delta^{(t-1)}  \texttt{ADec}(\addkey, x) \mod m$ 
        \State $\mu^{**} = \floor{y / \delta^{(t-1)}} \mod q$
        \vspace{1mm}
        \State $ \mu'' = \lfloor \mu^{**}/\Delta\rceil$
        \Comment{$\Delta=\round{q/p}$}
        \vspace{2mm}
        \State \Return $\mu'' \in \ZZ_{p}$
   \EndProcedure
    \end{algorithmic}
\end{algorithm}

\section{Proof of RLWE Compression}
\label{sec:prove-rlwe-compress}

\begin{proof}

Line 1 of \Cref{alg:rlwe-compress-response} computes 
$$
B[k] + \sum_{i=0}^{k} (q-A[k-i]) \cdot S[i] + \sum_{i=k+1}^{N-1} A[N+k-i] \cdot S[i]\
$$
encrypted under additive encryption, which is possible due to the linear properties.
We know that all coefficients of $A(X)$, $B(X)$, and $S(X)$ are elements in $\ZZ_q$, hence

{\footnotesize
    \begin{align*}
        B[k] + \left(\sum_{i=0}^{k} (q-A[k-i]) \cdot S[i]\right)
        + \left(\sum_{i=k+1}^{N-1} A[N+k-i] \cdot S[i]\right)\\
        \leq q + \left(\sum_{i=0}^{k} q \cdot q\right) + \left(\sum_{i=k+1}^{N-1} q \cdot q\right)
        = q + Nq^2 < m
    \end{align*}
}%
so there is no overflow in the plaintext space of the additive cryptosystem.

{\tiny
    \begin{align*}
        & \mu^{**}_{k} = \texttt{ADec}_{s}(x) \mod q \\
        & = \left(\left(B[k] + \sum_{i=0}^{k} (q-A[k-i]) \cdot S[i] \right. \right.\\
        & ~~~~~~~~~~~~~~~~~~~~~~~~~~ \left.\left.+\sum_{i=k+1}^{N-1} A[N+k-i] \cdot S[i]\right) \mod m \right) \mod q \\
        & = \left( B[k] + \sum_{i=0}^{k} (q-A[k-i]) \cdot S[i] + \sum_{i=k+1}^{N-1} A[N+k-i] \cdot S[i]\right) \mod q \\
        & = B[k] - \sum_{i=0}^{k} A[k-i] \cdot S[i] + \sum_{i=k+1}^{N-1} A[N+k-i] \cdot S[i] \mod q
    \end{align*}
}%

which is equivalent to the $k^{th}$ coefficient of 
$$
    \mu^*(X) = B(X) - A(X) \cdot S(X) \mod R_q
$$
which can be seen by expanding the equation.
Given that line 16 of \Cref{alg:lwe-encrypt-decrypt} performs rounding coefficient-wise, it produces the same result as line 10 of \Cref{alg:rlwe-compress-response}.
\end{proof}

\section{Modulus Switching Theorem}
\label{sec:modulus-switching-theorem}

Modulus switching is a known technique, mainly for RLWE schemes for various reasons.
We use it as well but for LWE-based schemes.
Particularly, if the error in an LWE ciphertext is not too high, the modulus can be smaller. The following theorem summarizes this fact.

\begin{lemma}
    Define LWE dimension $n$, ciphertext modulus $q$, plaintext modulus $p$, such that $p|q$.
    Also, assume $\sk\in\{0,1\}^{n}$ is a binary secret key.
    Assume $ct=(\textbf{a},b)$ is an LWE ciphertext encrypting message $m\in \ZZ_p$ such that $b = \sum \textbf{a}[i] \sk[i] + e + \frac{q}{p}m \mod q$ such that $|e|<\frac{q}{4p}$.
    Now define $ct' = (\textbf{a}',b')$ where $\textbf{a}'[i]=\round{\frac{r\textbf{a}[i]}{q}}$ and $b'=\round{\frac{r b}{q}}$. If $r\geq 2(n+1)p$ and then 
    {\footnotesize
    \begin{align}
        \textsc{LWEDecrypt}_{r,p}(\sk, ct) = \textsc{LWEDecrypt}_{q,p}(\sk, ct') = m
    \end{align}
    }
\end{lemma}

\newcommand{\half}{\frac{1}{2}}

\begin{proof}
    By definition, we can find a $k\in\ZZ$ such that
    \begin{align}
        b - kq = \sum \textbf{a}[i] \sk[i] + e + \frac{q}{p} m
    \end{align}
    So
    \begin{align}
        & b' - \sum \textbf{a}[i]' \cdot \sk[i] \\
        &= \round{\frac{r b}{q}} - \sum \round{\frac{r\textbf{a}[i]}{q}} \cdot \sk[i]  \\
        &\leq \frac{r b}{q} + \half - \sum (\frac{r\textbf{a}[i]}{q} - \half) \cdot \sk[i] \\
        &= \frac{r}{q}(b' - \sum \textbf{a}[i]\sk[i]) + \half + \half \sum \sk[i] \\
        &= \frac{r}{q} ( kq + e + \frac{q}{p} m) + \half + \half \sum \sk[i] \\
        &= kr + \frac{r}{p} m + \frac{er}{q} + \half + \half \sum \sk[i]
    \end{align}
    Now if we name $e' = \frac{er}{q} + \half + \half \sum \sk[i]$ then it suffices to have $|e'| < \frac{r}{2p}$ to have correct decryption, and using the assumptions we can see that 
    {\footnotesize
    \begin{align}
        |e'| = |\frac{er}{q} + \half + \half \sum \sk[i]| \leq |\frac{q}{4p}\frac{r}{q}| + |\frac{n+1}{2}| \leq |\frac{r}{2p}|
    \end{align}
    }
\end{proof}

\section{Definitions}
Our compression technique and proposed PIR protocol rely on two important problems, Learning with Errors and the security of the Paillier cryptosystem.

\subsection{Learning with Errors}
The Learning with Errors (LWE) assumption is parameterized by dimension $n\in\NN$, modulus $q\in\NN$, an error distribution $\chi_e$ over $\ZZ$, and a secret key distribution $\chi_s$ over $\ZZ_q$, and number of samples $m\in\NN$. 
The LWE hardness assumption states that the following two distributions
\begin{align*}
    \mathcal{D}_0 &= \{ (\mask, \mask\cdot\sk+e ) | \mask\sample\ZZ_{m\times n},~\sk\leftarrow\chi_s,~e\leftarrow\chi_e^m\} \\
    \mathcal{D}_1 &= \{ (\mask, r ) | \mask\sample\ZZ_{m\times n},~r\sample\ZZ_q^m\}
\end{align*}
are computationally indistinguishable. 
More concretely, $(n,q,\chi_e,\chi_s)$-LWE is $\epsilon$-hard, if for any PPT time adversary $\adv$ has at most $\epsilon$ advantage in distinguishing the two distributions.

\subsection{Paillier Cryptosystem}
The Paillier cryptosystem~\cite{paillierPublicKeyCryptosystemsBased1999} is a semantically secure cryptosystem based on the hardness of the composite residuosity assumption.
We say that Paillier is $\epsilon_{P}$-secure, if any PPT adversary has at most $\epsilon_{P}$ advantage in the IND-CPA game.

\section{Security \& Correctness of ZipPIR}
\label{appendix:proof}

\begin{theorem}[Correctness]
    For LWE parameters $(n,q,\chi)$ where $\chi$ is a discrete Gaussian with standard deviation $\sigma$, plaintext modulus $p$ such that $p|q$, and failure rate $\delta$ such that 
    \begin{align}
        \Delta > 2p\sigma \sqrt{2d_0\ln (2/\delta)} 
    \end{align}
    where $\Delta=q/p$ and for random $\mask\in\ZZ_q^{d_0\times n}$, for any database $\db\in\ZZ_{p}^{N_0\times d_0}$, and any query $(i,i_0)\in[N_0]\times[d_0]$, if 
    \begin{align}
        \hint\leftarrow\textsc{Setup}(\db) \\
        (\paillierkey, \qu) \leftarrow \textsc{Query}((i,i_0)) \\
        \ans \leftarrow \textsc{Response}(\db, \qu) \\
        f \leftarrow \textsc{Extract}(\paillierkey, \ans)
    \end{align}
    then $\PP[\db[i][i_0] = f] > 1 - \delta$.
\end{theorem}

\begin{proof}
We break down the proof into several steps.
First, we prove that $\textsc{LWEDecrypt}_{q,p}(\sk, D[j]) = \db[j][i_0]$ for every $j\in[N_0]$ with probability $\delta$.
We see that 
\begin{align}
    D[j] & = [H[j]|b[j]] \\
    & = [\db\cdot\mask[j] | \db\cdot \qu_0[j] ]\\
    & = (\db \cdot [\mask|\qu_0])[j]\\
    & = (\db[j] \cdot [\mask|\qu_0])
\end{align}
Now, note that in the first step of LWE decryption, we compute the following
\begin{align}
    & D[j][n] - D[j][:n] \cdot \sk \\
    & = \db[j] \cdot \qu_0 - \db[j] \cdot \mask \cdot \sk \\
    & = \db[j] \cdot (\mask \cdot \sk + e + \Delta u_0) - \db[j] \cdot \mask \cdot \sk \\
    & = \db[j] \cdot e + \db[j] \cdot \Delta u_0 = e_0 + \Delta \db[j][i_0]
\end{align}
where we define $e_0=\db[j] \cdot e$.
Note the required bound for correct decryption is simply $|e_0|<\Delta/2$, but we check for a stricter condition which will become necessary in the next step.
Specifically, we will check the probability that $|e_0|<\Delta/4$. Assume that the standard deviation of the discrete Gaussian distribution $\chi$ is $\sigma=\frac{s}{2\pi}$ for some $s>0$.
We follow the analysis of Henzinger et al.\cite[Theorem C.1]{henzingerOneServerPrice2023} and we see that for any $T>0$,  
\begin{align}
    \PP\left[ | \db[j] \cdot e | \geq T s \eudist{\db[j]} \right] < 2\exp(-\pi T^2)
\end{align}
where $\eudist{\cdot}$ denotes the Euclidean norm.
We plug in $T = \frac{\Delta}{4s\eudist{\db[j]}}$ and also observe that $\eudist{\db[j]} \leq \sqrt{d_0 (p/2)^2}$. Simplifying the equation and plugging in the fact that $s=\sigma\sqrt{2\pi}$, we see that 

\begin{align}
    \PP\left[ | \db[j] \cdot e | \geq \Delta/4 \right] &< 2\exp(-\pi T^2) \\
    &= 2\exp(-\pi (\frac{\Delta }{ 4s\eudist{\db[j]}}) ^2) \\
    &\leq 2\exp(-\pi (\frac{\Delta}{4\sigma \sqrt{2\pi} \frac{p}{2}\sqrt{d_0}})^2) \\
    &\leq 2\exp(-(\frac{\Delta}{2p \sigma\sqrt{2 d_0}})^2)
\end{align}

so we only require the last equation to be less than $\delta$ for the theorem to hold, which will occur if and only if

\begin{align}
    2\exp(-(\frac{\Delta}{2p \sigma\sqrt{2 d_0}})^2) < \delta
    \iff \Delta > 2p\sigma \sqrt{2d_0\ln (2/\delta)} 
\end{align}

So given that $|e_0| < \Delta/4 < \Delta/2$ with high probability, the decryption of $D[j]$ will succeed with high probability.

In the next step of the proof, we exploit the fact that $|e_0| < \Delta/4$ with high probability to use the modulus switching theorem from \appsection{sec:modulus-switching-theorem}. Hence, with probability $1-\delta$, 
\begin{align}
    \textsc{LWEDecrypt}(\sk, D_0[j]) = \textsc{LWEDecrypt}(\sk, D[j])
\end{align}

In the next step, we prove that for $\ell\in[k+1]$
\begin{align}
    f_{\ell} = c_{\ell}[i_{\ell+1}N_{\ell+1} + \cdots + i_{k}N_{k}]
    \label{eq:f-to-c-transform}
\end{align}
which we prove by induction. By definition, we know that $f_{k} = c_k[0]$. 
Assume that \Cref{eq:f-to-c-transform} holds for every $\ell \geq \ell'$ for some $\ell'$. We will show that \Cref{eq:f-to-c-transform} holds for $\ell=\ell'-1$.
For succinctness, define $t=i_{\ell+1}N_{\ell+1} + \cdots + i_{k}N_{k}$ and $\dec(\cdot)=\textsc{PaillierDecrypt}(\paillierkey, \cdot)$. Now for any $h\in[2^{\ell}]$ we have
\begin{align}
    f_{\ell-1}[h]
    & = m \cdot p_{\ell}[2h+1] + p_{\ell}[2h] \\
    & = m \cdot \dec(f_{\ell}[2h+1]) + \dec(f_{\ell}[2h]) \\
    & = m \cdot \dec(c_{\ell}[t][2h+1]) +\dec(c_{\ell}[t][2h]) \\
    & = m \cdot \floor{c_{\ell-1}[i_{\ell}N_{\ell} + t][h]/m} \\
    &   + c_{\ell-1}[i_{\ell}N_{\ell}+t][h] \mod m \\
    & = c_{\ell-1}[i_{\ell}N_{\ell}+t][h]\\
    & = c_{\ell-1}[i_{\ell}N_{\ell}+i_{\ell+1}N_{\ell+1} + \cdots + i_{k}N_{k}][h]
\end{align}
and by combining this for all values of $h$, we see that 
$$
    f_{\ell-1} = c_{\ell-1}[i_{\ell}N_{\ell} + \cdots + i_{k}N_{k}]
$$
We now have all the pieces to prove the full theorem. 
Due to the last step, we can see that $f_{0} = c_{0}[i_{0}N_{0} + \cdots + i_{k}N_{k}]=c_0[i]$.
Note that we used the fact that based on the definition of $i_\ell$, we have $i=i_{0}N_{0} + \cdots + i_{k}N_{k}$.
So  
\begin{align}
   f & = \textsc{ModifiedLWEDecrypt}_{r, p}(\paillierkey, f_0) \\
     & = \textsc{ModifiedLWEDecrypt}_{r, p}(\paillierkey, c_0[i]) \\
     & = \textsc{LWEDecrypt}_{r, p}(\sk, \textbf{D}_0[i]) \\
     & = \textsc{LWEDecrypt}_{q, p}(\sk, \textbf{D}[i]) \\
     & = \db[i][i_0]
\end{align}
where the second to third line holds since
\begin{align*}
    c_0[i] = \textsc{FastLWECompress}_{r}(\eck, \textbf{D}_0[i])
\end{align*}
and this proves the theorem.
\end{proof}

\end{appendices}

\end{document}